\documentclass{svjour2}                    
\smartqed  
\usepackage{graphicx}
 \usepackage{mathptmx}      
\usepackage{amsmath,amssymb}
\newcommand{\R}{\mathbb{R}}

\newcommand{\Exp}{\mathbb{E}}
\newcommand{\dd}{\mathrm{d}}
\newcommand{\ee}{\mathrm{e}}
\newcommand{\dto}{\downarrow}
\newcommand{\uto}{\uparrow}

\DeclareMathOperator{\tr}{tr}
\DeclareMathOperator{\dive}{div}

\DeclareMathOperator*{\mysimeq}{\sim}
\newcommand{\OrdExp}[1]{\left\lfloor \exp\left(#1\right)\right\rfloor}
\newcommand{\ordexp}[1]{\lfloor \exp(#1)\rfloor}

\begin{document}

\title{Path integral derivation and numerical computation of large deviation prefactors for non-equilibrium dynamics through matrix Riccati equations
}

\titlerunning{NESS prefactors and matrix Riccati equations}        

\author{Freddy Bouchet         \and
        Julien Reygner
}


\institute{F. Bouchet \at
              Univ Lyon, Ens de Lyon, Univ Claude Bernard, CNRS, Laboratoire de Physique, Lyon, France  \\
              \email{Freddy.Bouchet@ens-lyon.fr}           
           \and
           J. Reygner \at
              CERMICS, Ecole des Ponts, Marne-la-Vallée, France\\
              \email{julien.reygner@enpc.fr}
}

\date{}

\maketitle

\begin{abstract}

For many non-equilibrium dynamics driven by small noise, in physics, chemistry, biology, or economy, rare events do matter. Large deviation theory then explains that the leading order term of the main statistical quantities have an exponential behavior. The exponential rate is often obtained as the infimum of an action, which is minimized along an instanton. In this paper, we consider the computation of the next order sub-exponential prefactors, which are crucial for a large number of applications. Following a path integral approach, we derive the dynamics of the Gaussian fluctuations around the instanton and compute from it the sub-exponential prefactors. As might be expected, the formalism leads to the computation of functional determinants and matrix Riccati equations. By contrast with the cases of equilibrium dynamics with detailed balance or generalized detailed balance, we stress the specific non locality of the solutions of the Riccati equation: the prefactors depend on fluctuations all along the instanton and not just at its starting and ending points. We explain how to numerically compute the prefactors. The case of statistically stationary quantities requires considerations of non trivial initial conditions for the matrix Riccati equation. 

\end{abstract}

\keywords{Non-equilibrium statistical physics  \and Rare events \and Large deviation theory \and  Arrhenius law \and  Sub-exponential prefactors \and Stochastic differential equations}


\section{Introduction}

\subsection{Rare events, instantons and sub-exponential prefactors} Many systems in physics, chemistry, economics or biology can be described by stochastic differential equations with small noise. In such cases many statistical quantities, for instance the invariant distribution, first exit times, mean first passage times or transition probabilities have asymptotic exponential behavior $C^\epsilon\exp(-I/\epsilon)$, where $I$ is the exponential rate, $\epsilon$ is the noise amplitude, and $C^\epsilon$ is a sub-exponential prefactor (see below for a more precise definition). 

This mathematical remark has profound consequences in physics. The most classical examples of such exponential rates are the Arrhenius law, or thermodynamic potentials\footnote{Thermodynamic potentials are usually defined as static properties, independently of the dynamics, but they also appear as quasipotential in effective dynamical theory, a classical example being macroscopic fluctuation theory~\cite{BerDeSGabJonLan15}.}. Besides static properties, from a dynamical perspective, when conditioned on the occurence of a rare event, path probabilities often concentrate close to a predictable path, called instanton. This is a key and fascinating property for the dynamics of rare events and of their impact, which was first observed in statistical physics, for the nucleation of a classical supersaturated vapor~\cite{langer_1967_condensation_point}. Soon after, a similar concentration of path probabilities has been studied in gauge field theories~\cite{Coleman:1978ae,zinn1996quantum}, for instance for the Yang--Mills theory. Instantons continue to have number of applications in modern statistical physics, for instance to describe excitation chains at the glass transition~\cite{langer2006excitation}, reaction paths in chemistry~\cite{kampen_stochastic_2007}, escape of Brownian particles in soft matter~\cite{woillez2019escape}. 

The computation of the rate $I$ is the subject of classical techniques using Laplace asymptotics, for instance in classical or path integrals~\cite{Coleman:1978ae,zinn1996quantum,Gra88}. At the mathematical level, this is the subject of large deviation theory, see for instance the Freidlin--Wentzell theory for small noise large deviations~\cite{FreWen12}. However, for most genuine applications, computing the rate $I$ is not sufficient and a proper estimation of the sub-exponential prefactor $C^\epsilon$, or of its asymptotic behavior in the limit of small noise $\epsilon \dto 0$, is required. From a field theory perspective, such computations require the estimation of the path integrals at next to leading order. Such computations are very classical in the field theory context and involve the estimation of expectations over Gaussian processes, for which solutions of Riccati equations are needed. Several classical tools and approaches have been devised in quantum field theory and for equilibrium problems, often on a case by case basis (see for instance~\cite{Callan:1977pt,zinn1996quantum} as examples among many others).

Recently, rare events, instantons, transitions rates, have been studied in far from equilibrium systems and non-equilibrium steady states, were one starts from dynamics without detailed balance. The statistical mechanics approaches have then be extended to scientific fields so far unexpected. For instance rare events, instantons, and related concepts have been used in turbulence~\cite{grafke2013instanton,LAURIE:2015:A,grafke2015efficient,bouchet2019rare,dematteis2018rogue}, atmosphere dynamics~\cite{bouchet2019rare,simonnet2021multistability}, climate dynamics~\cite{RAGONE:2018:A}, astronomy~\cite{woillez2020instantons,abbot2021rare}, among many other examples. Moreover, a large effort has been pursued to develop dedicated numerical approches to compute the related instantons~\cite{grafke2019numerical}.
	
For all of these cases, it is essential to go beyond the computation of the exponential rates and to compute the sub-exponential prefactors. Then, although formally classical ideas still apply, several simplifications related to equilibrium dynamics no more occur. Often, the extent of possible analytic simplifications is limited, and one has to rely more on numerical simulations for actual computations. The aim of this paper is to develop the theoretical analysis at a level were it will be useful for devising numerical algorithms. We will thus consider stochastic differential equations with small noise, and develop the formalism and show the potential for numerical computations. The numerical computation will be illustrated on a simple example. 
	
At a technical level, we will derive the needed matrix Riccati equations, discuss their initial conditions which are not trivial, for instance in the case of statistically stationary quantities, and explain the relation between the matrix Riccati equations and the quantities of interest. 

\subsection{Setting} In this article, we consider systems described by the stochastic differential equation
\begin{equation}\label{eq:SDE}
  \dd X^{\epsilon}_s = b(X^{\epsilon}_s)\dd s + \sqrt{2\epsilon} \dd W_s, \qquad s \geq 0,
\end{equation}
in $\R^d$, where $b : \R^d \to \R^d$ is a smooth vector field, $W$ is a $d$-dimensional Brownian motion, and $\epsilon > 0$ can be interpreted as a temperature parameter. In general, the diffusion process $X^\epsilon$ defined by~\eqref{eq:SDE} is not reversible, so that it serves as a model for physical systems driven out of equilibrium. In particular, when the process possesses a stationary distribution $P^\epsilon$ --- called the \emph{non-equilibrium steady state}, the latter may not be explicit. Still, the Freidlin--Wentzell theory~\cite{FreWen12} asserts that in the $\epsilon \dto 0$ limit and under suitable assumptions which we shall detail in Section~\ref{s:C} below, $P^\epsilon$ satisfies a large deviation principle, with a rate function $V$ called the \emph{quasipotential} and defined in~\eqref{eq:V}. We shall formally denote by
\begin{equation}\label{eq:FW}
  \lim_{\epsilon \dto 0} -\epsilon \log P^{\epsilon}(x) = V(x)
\end{equation}
this large deviation principle, and we refer to~\cite{DemZei10} for standard material on large deviation theory. The quasipotential is known to play the role of an entropy function for non-equilibrium models~\cite{FreWen12,Gra88,BerDeSGabJonLan15}. 

Let us define the large deviation \emph{prefactor} $C^\epsilon(x)$ to the non-equilibrium steady state $P^\epsilon$ by the identity
\begin{equation}\label{eq:Cepsilon}
  P^{\epsilon}(x) = C^\epsilon(x)\exp\left(-\frac{V(x)}{\epsilon}\right).
\end{equation}
An equivalent formulation of~\eqref{eq:FW} is the assertion that $\lim_{\epsilon \dto 0} \epsilon \log C^\epsilon(x) = 0$, so that no precise information on the prefactor can be obtained merely from large deviation theory. However, combining notions from this theory with a WKB approximation, we derived an asymptotic equivalent of the prefactor in~\cite{BouRey16}, recalled in Equation~\eqref{eq:equivC} below. 

The purpose of the present article is to continue our study of the prefactor, by proposing an alternative method to obtain~\eqref{eq:equivC}, based on the path integral formulation, and presenting a numerical method to effectively compute the terms appearing in this formula. Both tasks rely on the study of matrix-valued Riccati equations satisfied by quantities related to the fluctuations of the diffusion process $X^\epsilon$ around given deterministic paths.


\subsection{Organization of the paper}

In Section~\ref{s:C}, we state the asymptotic equivalent of the prefactor obtained in~\cite{BouRey16} and recall the notions of large deviation theory involved in this formula. In Section~\ref{s:path}, we describe an alternative method leading to the same formula, which is based on the path integral formulation of the non-equilibrium steady state and establishes a connection between the prefactor and the fluctuations of the process $X^\epsilon$ around its most probable path. The relation with sharp asymptotics for mean exit times is discussed in Section~\ref{s:exit}. Section~\ref{s:num} is dedicated to the numerical computation of the asymptotic equivalent of the prefactor, which we apply on an illustrative example in Section~\ref{s:na}. The conclusive Section~\ref{s:conc} summarizes the main contribution of the paper and discusses their relation with other recent works.


\subsection{Notations}\label{ss:not} The Euclidian norm and inner product of $\R^d$ are respectively denoted by $\|\cdot\|$ and $\langle \cdot, \cdot\rangle$. We write $I_d$ for the identity matrix of size $d \times d$. Given a smooth function $f : \R^d \to \R$ (typically, $V$ or $b_k$, $\ell_k$), we denote by $\nabla f$ and $\nabla^2 f$ the gradient and Hessian matrix of $f$. Similarly, given a smooth vector field $F : \R^d \to \R^d$ (typically $b$ or $\ell$), $\nabla F(x)$ refers to the matrix with coefficient $\partial_j F_i(x)$ at the $i$-th row and $j$-th column (that is to say, $\nabla F$ is the Jacobian matrix of $F$), and $\nabla^2F(x)(y,y)$ refers to the vector whose $i$-th coordinate is $\langle y, \nabla^2 F_i(x) y\rangle$. The symbols $\Delta$ and $\dive$ refer to the usual Laplacian and divergence operators of scalar functions. The closure of an open set $D \subset \R^d$ is denoted by $\bar{D}$. For two positive quantities $u_\epsilon$, $v_\epsilon$ indexed by $\epsilon>0$, the notation $u_\epsilon \mysimeq_{\epsilon \dto 0} v_\epsilon$ means that the ratio $u_\epsilon/v_\epsilon$ converges to $1$ when $\epsilon \dto 0$.


\section{Prefactor for the stationary distribution}\label{s:C}

In this section, we first precise the assumptions under which we shall work. We then provide a very brief summary of the main notions from the Freidlin--Wentzell theory on which we shall rely. Finally, we provide an asymptotic equivalent for the prefactor $C^\epsilon(x)$, obtained in~\cite{BouRey16} through a WKB approximation.


\subsection{Assumptions on the vector field $b$} Throughout this section, we make the following assumptions:
\begin{itemize}
  \item[(A1)] the deterministic system $\dot{x}=b(x)$ possesses a unique equilibrium point $\bar{x} \in \R^d$, which attracts all the trajectories,
  \item[(A2)] for all $\epsilon>0$, the SDE~\eqref{eq:SDE} possesses a unique stationary measure $P^{\epsilon}$.
\end{itemize}

\begin{remark}
  Assumption~(A1) allows to streamline the exposition of our results. In the more general case where the deterministic system $\dot{x}=b(x)$ possesses several (isolated) equilibrium points, our arguments can be adapted and produce statements that are valid in the neighborhood of the equilibrium points.
\end{remark}


\subsection{Action functional, quasipotential and Hamilton--Jacobi equation}\label{ss:FW}

We recall a few notions from the Freidlin--Wentzell theory~\cite[Chapter~4]{FreWen12} that will be useful in the paper.

The \emph{action functional} for the stochastic differential equation~\eqref{eq:SDE} is defined for a trajectory $\phi = (\phi_s)_{t_1 \leq s \leq t_2}$ on the time interval $[t_1,t_2]$ by
\begin{equation}
  \mathcal{A}_{t_1,t_2}[\phi] = \begin{cases}
    \int_{s=t_1}^{t_2} \mathcal{L}(\phi_s, \dot{\phi}_s) \dd s & \text{if $\phi$ is absolutely continuous},\\
    +\infty & \text{otherwise},  \end{cases}
\end{equation}
where the Lagrangian $\mathcal{L}$ writes
\begin{equation}
  \forall (x,v) \in \R^d \times \R^d, \qquad \mathcal{L}(x,v) = \frac{1}{4}\|v-b(x)\|^2.
\end{equation}
It describes the large deviations of the trajectory $(X^{\epsilon}_s)_{s \in [t_1,t_2]}$ when $\epsilon \dto 0$~\cite[Theorem~1.1, p.~86]{FreWen12}.

The \emph{quasipotential} with respect to $\bar{x}$ is defined by the variational formula
\begin{equation}\label{eq:V}
  V(x) = \inf\{\mathcal{A}_{t_1,t_2}[\phi], \phi_{t_1} = \bar{x}, \phi_{t_2} = x, t_1 < t_2\},
\end{equation}
where the infimum runs over all finite time intervals $[t_1,t_2]$ and trajectories $\phi = (\phi_s)_{t_1 \leq s \leq t_2}$. The quasipotential is nonnegative. Besides, it is known that if it is smooth, then it solves the Hamilton--Jacobi equation
\begin{equation}\label{eq:HJ}
  \langle \nabla V(x), \nabla V(x)\rangle + \langle b(x), \nabla V(x)\rangle = 0,
\end{equation}
in domains of $\R^d$ in which, in particular, $\bar{x}$ is the only critical point of $V$. We refer to the discussion in~\cite[pp.~100--101]{FreWen12} for details. In this work, we shall assume that this property actually holds globally:
\begin{itemize}
  \item[(A3)] the quasipotential $V$ is $C^2$ on $\R^d$, $\nabla V(x) \not=0$ for any $x \not= \bar{x}$, and the Hamilton--Jacobi equation~\eqref{eq:HJ} holds on $\R^d$.
\end{itemize}
Equivalently, the vector field $\ell$ defined by 
\begin{equation}\label{eq:decomp}
  b(x) = -\nabla V(x) + \ell(x)
\end{equation}
satisfies
\begin{equation}\label{eq:transv}
  \forall x \in \R^d, \qquad \langle \nabla V(x), \ell(x)\rangle = 0.
\end{equation}
Assumption~(A3) then implies that the quasipotential satisfies the identity
\begin{equation}\label{eq:VMAP}
  V(x) = \int_{s=-\infty}^0 \mathcal{L}(\varphi^x_s,\dot{\varphi}^x_s)\dd s = \mathcal{A}_{-\infty,0}[\varphi^x],
\end{equation}
where $\varphi^x = (\varphi^x_s)_{s \leq 0}$ is the minimum action path joining $\bar{x}$ to $x$, defined as the unique solution to the backward Cauchy problem
\begin{equation}\label{eq:varphi}
  \left\{\begin{aligned}
    & \dot{\varphi}^x_s = \nabla V(\varphi^x_s) + \ell(\varphi^x_s), \quad s \leq 0,\\
    & \varphi^x_0 = x, \quad \lim_{s \to -\infty} \varphi^x_s = \bar{x}.
  \end{aligned}\right.
\end{equation}
The identity~\eqref{eq:VMAP} shows in particular that in~\eqref{eq:V}, the quantity $V(x)$ could be equivalently defined as $\inf\{\mathcal{A}_{-\infty,0}[\phi], \lim_{t \to -\infty} \phi_t = \bar{x}, \phi_0 = x\}$. In the sequel of the paper, we shall call $\varphi^x$ the \emph{fluctuation path}, as it describes the most probable path followed by a fluctuation of the diffusion process $X^\epsilon$ joining $\bar{x}$ to $x$.

\begin{remark}
  Under Assumption~(A3), differentiating the equality $\langle \nabla V, \ell\rangle = 0$ twice, we get the identity
  \begin{equation}\label{eq:D2transv}
    \sum_{k=1}^d (\partial_k \nabla^2 V) \ell_k + \nabla^2 V \nabla \ell + \nabla \ell^{\top} \nabla^2 V + \sum_{k=1}^d (\partial_k V) \nabla^2 \ell_k = 0,
  \end{equation}
  which shall be useful in the course of the paper. We recall here that the notation $\nabla^2$ refers to the Hessian matrix, see Section~\ref{ss:not}.
\end{remark}

\begin{remark}
  Putting~\eqref{eq:decomp} and~\eqref{eq:varphi} together shows that the fluctuation path, the quasipotential and the drift $b$ satisfy the identity
  \begin{equation}\label{eq:V-varphi-b}
    \nabla V(\varphi^x_s) = \frac{1}{2}\left(\dot{\varphi}^x_s-b(\varphi^x_s)\right),
  \end{equation}
  which will be used in the sequel of the paper.
\end{remark}

\begin{remark} Assumption~(A3) is satisfied in particular if the vector field $b$ is known to possess an explicit transverse decomposition of the form $b = -\nabla V + \ell$, for some smooth function $V : \R^d \to \R$ which is such that $\langle \nabla V, \ell \rangle = 0$ on $\R^d$, and $\nabla V(x) \not= 0$ if $x \not= \bar{x}$. In this case, $V$ can be then shown to coincide with the quasipotential with respect to $\bar{x}$, defined by the right-hand side of the identity~\eqref{eq:V}. However, even if the vector field $b$ is smooth, it is in general difficult to prove directly that the quasipotential, defined by the right-hand side of~\eqref{eq:V}, is smooth.
\end{remark}


\subsection{WKB derivation of the formula for the prefactor}\label{ss:intro-wkb}

Injecting both the decomposition~\eqref{eq:decomp} of the vector field $b$ and the ansatz~\eqref{eq:Cepsilon} in the stationary Fokker--Planck equation
\begin{equation}\label{eq:sFP}
  \epsilon \Delta P^\epsilon - \dive(b P^\epsilon) = 0
\end{equation}
shows that $P^\epsilon$ is equal to the Gibbs measure
\begin{equation}\label{eq:Gibbs}
  P^\epsilon_{\mathrm{Gibbs}}(x) = \frac{1}{Z^\epsilon} \exp\left(-\frac{V(x)}{\epsilon}\right), \qquad Z^\epsilon = \int_{x \in \R^d} \exp\left(-\frac{V(x)}{\epsilon}\right)\dd x,
\end{equation}
if and only if the condition
\begin{equation}\label{eq:divl}
  \forall x \in \R^d, \qquad \dive \ell(x) = 0
\end{equation}
holds. In this case, the prefactor $C^\epsilon$ does not depend on $x$ and the Laplace approximation for $Z^\epsilon$ provides the asymptotic equivalence
\begin{equation}
  C^\epsilon \mysimeq_{\epsilon \dto 0} \sqrt{\frac{\det \nabla^2 V(\bar{x})}{(2\pi\epsilon)^d}},
\end{equation}
as soon as the following assumption holds:
\begin{itemize}
  \item[(A4)] the matrix $\nabla^2 V(\bar{x})$ is positive-definite.
\end{itemize}

If the condition~\eqref{eq:divl} does not hold, an expansion in powers of $\epsilon$ can be performed in~\eqref{eq:sFP}, following the usual WKB approximation method. This leads to the prefactor equivalent
\begin{equation}\label{eq:equivC}
  C^\epsilon(x) \mysimeq_{\epsilon \dto 0} \sqrt{\frac{\det \nabla^2 V(\bar{x})}{(2\pi\epsilon)^d}}\exp\left(-\int_{s=-\infty}^0 \dive\ell(\varphi^x_s)\dd s\right),
\end{equation}
which was derived in~\cite[Section~3]{BouRey16}, see also previous results in~\cite{CohLew67,Lud75,MaiSte97,Sch09}. The supplementary exponential term appearing in the right-hand side of~\eqref{eq:equivC} depends on the value of $\dive\ell$ along the whole trajectory of the fluctuation path $\varphi^x$. Therefore, we shall call it the \emph{nonlocal contribution} to the prefactor.


\section{Derivation of the formula from the path integral formulation}\label{s:path}

In this section, we still work under Assumptions~(A1--4), and describe an alternative method to the WKB approximation in order to derive the formula~\eqref{eq:equivC} for the prefactor to the stationary distribution $P^\epsilon$. The method is based on the path integral formulation of $P^\epsilon$, in which a Laplace approximation is performed, leading to the computation of an infinite-dimensional Gaussian integral. The latter involves the fluctuations of the process $X^\epsilon$ around the path $\varphi^x$, and thanks to the Feynman--Kac formula, it is computed by solving a backward matrix Riccati equation.

Our derivation is rather formal; in particular, the Laplace approximation in the path integral formulation is not rigorously justified. However, and although it is certainly less straightforward than the direct WKB approximation sketched in the previous section (which is also nonrigorous), the method presented here has the conceptual advantage to establish a connection between the prefactor, written under the form of a functional determinant, and the process of fluctuations of $X^\epsilon$ around the path $\varphi^x$. The latter process was recently studied in the mathematical literature~\cite{LuStuWeb16,SanStu16}.

\subsection{Laplace approximation in the path integral formulation}\label{ss:laplace} In this subsection, we fix $X^{\epsilon}_0 = \bar{x}$. Then, observables of the trajectory $(X^\epsilon_s)_{s \in [0,t]}$ write in the path integral formalism
\begin{equation}
  \mathbb{E}\left[F\left((X^\epsilon_s)_{s \in [0,t]}\right)\right] = \int_{\phi_0=\bar{x}} F[\phi]
  \ee^{-\frac{1}{\epsilon}\mathcal{A}_{0,t}[\phi]}
  \mathcal{D}[\phi],
\end{equation}
where the integral is taken over all absolutely continuous, $\R^d$-valued trajectories $\phi=(\phi_s)_{s \in [0,t]}$ such that $\phi_0=\bar{x}$. This fact can be obtained following the standard time-discretization construction of path integrals, using Ito's convention for the discretization. Alternatively, with a more probabilistic point of view, it may be derived from the Girsanov theorem, once one takes the convention to denote by 
$\ee^{-\frac{1}{4\epsilon}\int_{s=0}^t \|\dot{\phi}_s\|^2 \dd s}\mathcal{D}[\phi]$ 
the law of the Brownian trajectory $(\bar{x}+\sqrt{2\epsilon}W_s)_{s \in [0,t]}$. 

In this formalism, we deduce that the density of the random variable $X^\epsilon_t$ writes
\begin{equation}
  P^\epsilon_t(x) = \int_{\phi_0=\bar{x}} \delta_0(\phi_t-x)
  \ee^{-\frac{1}{\epsilon}\mathcal{A}_{0,t}[\phi]}
  \mathcal{D}[\phi],
\end{equation}
for any $x \in \R^d$, where $\delta_0$ is the Dirac distribution at $0$. In the sequel, it is more convenient to consider trajectories defined on $[t,0]$, $t<0$, rather than on $[0,t]$, $t>0$. Therefore we introduce the notation
\begin{equation}
  \bar{P}^\epsilon_t(x) = P^\epsilon_{-t}(x) = \int_{\phi_t=\bar{x}} \delta_0(\phi_0-x)
  \ee^{-\frac{1}{\epsilon}\mathcal{A}_{t,0}[\phi]}
  \mathcal{D}[\phi],
\end{equation}
for $t<0$. As a consequence, the stationary distribution $P^{\epsilon}$ writes
\begin{equation}
  P^{\epsilon}(x) = \lim_{t \to +\infty} P^\epsilon_t(x) = \lim_{t \to -\infty} \bar{P}^\epsilon_t(x).
\end{equation}

The second-order expansion of the action functional on $[t,0]$ around the fluctuation path defined by~\eqref{eq:varphi} writes
\begin{equation}
  \mathcal{A}_{t,0}[\phi] \simeq \mathcal{A}_{t,0}[\varphi^x] + \frac{\delta \mathcal{A}_{t,0}}{\delta \phi}[\varphi^x](\phi-\varphi^x) + \frac{1}{2}\frac{\delta^2 \mathcal{A}_{t,0}}{\delta \phi^2}[\varphi^x](\phi-\varphi^x, \phi-\varphi^x).
\end{equation}
In the $t \to -\infty$ limit, the first-order term vanishes because of the optimality condition on $\varphi^x$. The second-order term rewrites
\begin{equation}
  \begin{aligned}
    & \frac{1}{2}\frac{\delta^2 \mathcal{A}_{t,0}}{\delta \phi^2}[\varphi^x](\delta \phi, \delta \phi)\\
    & = \frac{1}{4} \int_{s=t}^0 \left(\|\delta\dot{\phi}_s - \nabla b(\varphi^x_s) \delta\phi_s\|^2 - \langle \dot{\varphi}^x_s - b(\varphi^x_s), \nabla^2 b(\varphi^x_s)(\delta\phi_s, \delta\phi_s)\rangle\right)\dd s\\
    & = \frac{1}{4} \int_{s=t}^0 \left(\|\delta\dot{\phi}_s + Q^x_s \delta\phi_s\|^2 + 2\langle \delta\phi_s, R^x_s\delta\phi_s\rangle\right)\dd s,
  \end{aligned}
\end{equation}
where the matrices $Q^x_s$ and $R^x_s$ are defined, for $s \leq 0$, by
\begin{equation}\label{eq:QR}
  Q^x_s = -\nabla b(\varphi^x_s), \qquad R^x_s = -\sum_{k=1}^d \partial_k V(\varphi^x_s) \nabla^2 b_k(\varphi^x_s).
\end{equation}
At this stage, let us anticipate on the numerical discussion of Section~\ref{s:num}, and point out that the identity~\eqref{eq:V-varphi-b} shows that the matrices $Q^x_s$ and $R^x_s$ can be computed from the mere knowledge of the vector field $b$ with its space derivatives, and the fluctuation path $\varphi^x$ together with its time derivative. In particular, it is not necessary to compute neither the quasipotential nor its space derivatives along the fluctuation path.

As a consequence of the second-order expansion of the action functional, the Laplace approximation yields, for $\epsilon > 0$ small but fixed, 
\begin{equation}\label{eq:Peps:1}
  \bar{P}^{\epsilon}_t(x) \simeq 
  \ee^{-\frac{1}{\epsilon}\mathcal{A}_{t,0}[\varphi^x]}
  \int_{\delta\phi_t=\bar{x}-\varphi^x_t} \delta_0(\delta\phi_0)
  \ee^{-\frac{1}{4\epsilon}\int_{s=t}^0 \left(\|\delta\dot{\phi}_s + Q^x_s \delta\phi_s\|^2 + 2\langle \delta\phi_s, R^x_s\delta\phi_s\rangle\right)\dd s}
  \mathcal{D}[\delta\phi].
\end{equation} 
By~\eqref{eq:VMAP}, in the $t \to -\infty$ limit, $\mathcal{A}_{t,0}[\varphi^x]$ converges to the quasipotential $V(x)$, while~\eqref{eq:varphi} shows that $\bar{x}-\varphi^x_t$ vanishes. Therefore the prefactor $C^{\epsilon}(x)$ defined by~\eqref{eq:Cepsilon} is equivalent, when $\epsilon \dto 0$, to the $t \to -\infty$ limit of the path integral
\begin{equation}\label{eq:C}
   \begin{aligned}
     & \int_{\delta\phi_t=0} \delta_0(\delta\phi_0)
     \ee^{-\frac{1}{4\epsilon}\int_{s=t}^0 \left(\|\delta\dot{\phi}_s + Q^x_s \delta\phi_s\|^2 + 2\langle \delta\phi_s, R^x_s\delta\phi_s\rangle\right)\dd s}
     \mathcal{D}[\delta\phi]\\
     & = \frac{1}{\epsilon^{d/2}} \int_{\delta\phi_t=0} \delta_0(\delta\phi_0)
     \ee^{-\frac{1}{4}\int_{s=t}^0 \left(\|\delta\dot{\phi}_s + Q^x_s \delta\phi_s\|^2 + 2\langle \delta\phi_s, R^x_s\delta\phi_s\rangle\right)\dd s}
     \mathcal{D}[\delta\phi]
   \end{aligned}
\end{equation}
after rescaling $\delta\phi$ by the factor $\sqrt{\epsilon}$ and using the fact that $\delta_0(\sqrt{\epsilon}y) = \epsilon^{-d/2}\delta_0(y)$. Since the term appearing in the time integral above is quadratic in $\delta\phi$, the path integral is an infinite-dimensional Gaussian integral, the computation of which usually reduces to the computation of a functional determinant. In the next paragraph, we adopt a different strategy to compute its value, based on the use of the Feynman--Kac formula.


\subsection{Feynman--Kac formula and backward matrix Riccati equation}\label{ss:FK}

Let us define, for all $t \leq 0$ and $x,y \in \R^d$,
\begin{equation}\label{eq:u}
  u(x;t,y) = \int_{\delta\phi_t=y} \delta_0(\delta\phi_0)
  \ee^{-\frac{1}{4}\int_{s=t}^0 \left(\|\delta\dot{\phi}_s + Q^x_s \delta\phi_s\|^2 + 2\langle \delta\phi_s, R^x_s\delta\phi_s\rangle\right)\dd s}
  \mathcal{D}[\delta\phi],
\end{equation}
where the path integral is taken over trajectories $\delta\phi$ on $[t,0]$ such that $\delta\phi_t=y$,
so that by~(\ref{eq:Peps:1}--\ref{eq:C}),
\begin{equation}\label{eq:Cu}
  C^{\epsilon}(x) \mysimeq_{\epsilon \dto 0} \frac{1}{\epsilon^{d/2}} \lim_{t \to -\infty} u(x;t,0).
\end{equation}
The introduction of this notation allows to provide a probabilistic interpretation to the right-hand side of~\eqref{eq:u}. Indeed, the path integral there writes as the expectation
\begin{equation}\label{eq:u-PI}
  u(x;t,y) = \Exp_{t,y}\left[\exp\left(-\frac{1}{2}\int_{s=t}^0 \langle Y^x_s, R^x_sY^x_s\rangle\dd s\right)\delta_0(Y^{\epsilon}_0)\right],
\end{equation}
where $(Y^x_s)_{s \leq 0}$ is the diffusion process defined by
\begin{equation}\label{eq:Y}
  \dd Y^x_s = - Q^x_s Y^x_s \dd s + \sqrt{2} \dd W_s, \qquad s \leq 0,
\end{equation}
and $\Exp_{t,y}[\cdot]$ refers to the expectation under which $Y^x_t=y$ while $\delta_0$ still denotes the Dirac distribution in $0$. The linear process $(Y^x_s)_{s \leq 0}$ describes the scaled Gaussian fluctuations of the original process around $(\varphi^x_s)_{s \leq 0}$. 

By the Feynman--Kac formula, the function $u(x;t,y)$ is the solution to the backward parabolic problem
\begin{equation}\label{eq:FK}
  \left\{\begin{aligned}
    -\partial_t u(x;t,y) & = \Delta_y u(x;t,y) - \langle Q^x_t y, \nabla_y u(x;t,y)\rangle - \frac{1}{2} \langle y, R^x_t y\rangle u(x;t,y), \qquad t < 0,\\
    u(x;0,y) & = \delta_0(y).
  \end{aligned}\right.
\end{equation}
Since the diffusion process $(Y^x_s)_{s \leq 0}$ describes Gaussian fluctuations around the fluctuation path, we shall look for a solution to~\eqref{eq:FK} with the Gaussian ansatz
\begin{equation}\label{eq:ansatz}
  u(x;t,y) = \frac{1}{\sqrt{(4\pi)^d \eta^x_t}} \exp\left(-\frac{\langle y, K^x_t y\rangle}{4}\right),
\end{equation}
where, for all $t<0$, $\eta^x_t>0$ and the matrix $K^x_t$ is symmetric. This is a natural ansatz since~\eqref{eq:u-PI} shows that $u(x;t,\cdot)$ is the marginal distribution at time $s=0$ of the Gaussian process $(Y^\epsilon_s)_{s \in [t,0]}$, under an (unnormalized) exponential tilting by a quadratic functional of this process. Since such changes of measure are known to preserve Gaussian distributions, the resulting density may be expected to remain Gaussian. Furthermore, the condition $u(x;0,y) = \delta_0(y)$ implies the limit behavior
\begin{equation}\label{eq:limit}
  \lim_{t \uto 0} (K^x_t)^{-1} = 0, \quad \lim_{t \uto 0} \eta^x_t = 0, \quad \lim_{t \uto 0} \eta^x_t \det K^x_t = 1.
\end{equation}

Injecting the ansatz~\eqref{eq:ansatz} into the problem~\eqref{eq:FK}, we get the system of ordinary differential equations
\begin{equation}\label{eq:Riccati}
  \forall t < 0, \qquad \left\{\begin{aligned}
    \dot{\eta}^x_t & = - \eta^x_t \tr K^x_t,\\
    \dot{K}^x_t & = (K^x_t)^2 + {Q^x_t}^{\top}K^x_t + K^x_t Q^x_t - 2R^x_t.
  \end{aligned}\right.
\end{equation}
The second equation is a backward matrix Riccati equation, it is solved in Appendix~\ref{app:riccati}. We obtain the explicit result
\begin{equation}\label{eq:solK}
  \forall t < 0, \qquad K^x_t = -2\nabla^2 V(\varphi^x_t) + (Z^x_t)^{-1},
\end{equation}
with
\begin{equation}
  \begin{aligned}
    & Z^x_t = \int_{s=t}^0 \OrdExp{\int_{r=s}^t A^x_r \dd r}\OrdExp{\int_{r=s}^t A^x_r \dd r}^{\top}\dd s,\\
    & A^x_r = \nabla^2 V(\varphi^x_r) + \nabla \ell(\varphi^x_r),
  \end{aligned}
\end{equation}
and $\ordexp{\cdot}$ denotes the time ordered exponential of matrices, the definition and a few properties of which are recalled in Appendix~\ref{app:OE}. From~\eqref{eq:solK} we deduce in Appendix~\ref{app:pfeta} that
\begin{equation}\label{eq:limeta}
  \lim_{t \to -\infty} \eta^x_t = \frac{1}{2^d \det(\nabla^2 V(\bar{x}))}\exp\left(2\int_{s=-\infty}^0 \dive \ell(\varphi^x_s)\dd s\right).
\end{equation} 
Combining this result with the identities~\eqref{eq:Cu} and~\eqref{eq:ansatz} allows to recover the formula~\eqref{eq:equivC}. As a conclusive statement, we write the asymptotic equivalence for the non-equilibrium steady state
\begin{equation}\label{eq:asympPeps}
  P^\epsilon(x) \mysimeq_{\epsilon \dto 0} \sqrt{\frac{\det \nabla^2 V(\bar{x})}{(2\pi\epsilon)^d}}\exp\left(-\frac{V(x)}{\epsilon}-\int_{s=-\infty}^0 \dive\ell(\varphi^x_s)\dd s\right).
\end{equation}

The computation of the solution~\eqref{eq:solK} to the system~\eqref{eq:Riccati}, together with the proof of the identity~\eqref{eq:limeta}, make our derivation of the final asymptotics~\eqref{eq:asympPeps} for the non-equilibrium steady state definitely different from the WKB approximation from Section~\ref{s:C}. This chain of arguments may be considered as the main conceptual result from this article.


\section{Application to mean exit times}\label{s:exit}

Let $D \subset \R^d$ be a domain containing an equilibrium point $\bar{x}$ of the deterministic system $\dot{x}=b(x)$. Let us define the \emph{exit time} from $D$ by
\begin{equation}
  \tau^\epsilon_D = \inf\{s \geq 0 : X^\epsilon_s \not\in D\}.
\end{equation}
Under suitable assumptions on $D$, the Freidlin--Wentzell theory asserts that 
\begin{equation}\label{eq:Arrh}
  \lim_{\epsilon \dto 0} \epsilon \log\Exp_{\bar{x}}[\tau^\epsilon_D] = \inf_{y \in \partial D} V(y),
\end{equation}
where the subscript $\bar{x}$ indicates that $X^\epsilon_0 = \bar{x}$, and $V$ still refers to the quasipotential with respect to $\bar{x}$ defined by~\eqref{eq:V}. 

When the deterministic system $\dot{x}=b(x)$ possesses several stable equilibrium points and $D$ denotes the basin of attraction of one of these points $\bar{x}$, then the diffusion process $(X^\epsilon_s)_{s \geq 0}$ is called \emph{metastable}, and the logarithmic equivalent~\eqref{eq:Arrh} is called the \emph{Arrhenius law}.

In any case, the associated \emph{prefactor} $L^\epsilon_D$ to the mean exit time from $D$ is defined by
\begin{equation}
  \Exp_{\bar{x}}[\tau^\epsilon_D] = L^\epsilon_D \exp\left(\frac{1}{\epsilon} \inf_{y \in \partial D} V(y)\right).
\end{equation}
In this section, we express this prefactor in terms of the function $C^\epsilon(x)$ computed in~\eqref{eq:equivC}.

\subsection{Domain with a noncharacteristic boundary}\label{ss:nc} In this subsection, we assume that $D$ is an open, smooth and connected subset of $\R^d$ satisfying the following conditions, where we denote by $n(y)$ the exterior normal vector at $y \in \partial D$.
\begin{itemize}
  \item[(B1)] The deterministic system $\dot{x}=b(x)$ possesses a unique equilibrium point $\bar{x}$ in $D$, which attracts all the trajectories started from $D$, and $\langle b(y), n(y) \rangle < 0$ for all $y \in \partial D$.
  \item[(B2)] The function $V$ is $C^1$ in $D$; for any $x \in \bar{D}$, the fluctuation path $\varphi^x_s$ goes to $\bar{x}$ when $s \to -\infty$; and $\langle \nabla V(y), n(y)\rangle > 0$ for all $y \in \partial D$.
\end{itemize}
Under these assumptions, we proceed as in Subsection~\ref{ss:FW} and define the vector field $\ell : \bar{D} \to \R^d$ by the identity $b=-\nabla V+\ell$. Then $\ell$ still satisfies the orthogonality relation~\eqref{eq:transv}. 
\begin{itemize}
  \item[(B3)] The minimum of $V$ over $\partial D$ is reached at a single point $y^*$, at which
  \begin{equation}
    \mu^* = \langle \nabla V(y^*) + \ell(y^*), n(y^*) \rangle > 0,
  \end{equation}
  and the quadratic form $h^*: \xi \mapsto \langle \xi, \nabla^2 V(y^*) \xi\rangle$ has positive eigenvalues on the hyperplane $n(y^*)^\perp = \{\xi \in \R^d: \langle \xi, n(y^*)\rangle = 0\}$.
\end{itemize}

By Assumption~(B1), \cite[Theorem~4.1, p.~106]{FreWen12} applies and yields~\eqref{eq:Arrh}. In~\cite[Eq.~(4.25)]{BouRey16}, the integral formula 
\begin{equation}
  \lambda^\epsilon_D = \int_{y \in \partial D} \langle \nabla V(y) + \ell(y), n(y)\rangle C^\epsilon(y) \exp\left(-\frac{V(y)}{\epsilon}\right)\dd y
\end{equation}
was derived for the \emph{exit rate} $\lambda^\epsilon_D = \Exp_{\bar{x}}[\tau^\epsilon_D]^{-1}$. Using the second-order expansion of $V$ in the neighborhood of $y^*$ in this formula, we obtain the equivalent
\begin{equation}
  \begin{aligned}
    L^\epsilon_D & \mysimeq_{\epsilon\dto 0} \frac{1}{C^\epsilon(y^*) \mu^*} \sqrt{\frac{\det h^*}{(2\pi\epsilon)^{d-1}}}\\
    & \mysimeq_{\epsilon\dto 0} \frac{1}{\mu^*}\sqrt{\frac{2\pi\epsilon \det h^*}{\det \nabla^2 V(\bar{x})}}\exp\left(\int_{s=-\infty}^0 \dive \ell(\varphi^{y^*}_s)\dd s\right)
  \end{aligned}
\end{equation}
for the prefactor to the mean exit time from $D$.

\subsection{The Eyring--Kramers formula for metastable states}\label{ss:ek} In this subsection, we assume that the deterministic system $\dot{x}=b(x)$ possesses two stable equilibrium points $\bar{x}_1$ and $\bar{x}_2$, whose basins of attractions are separated by a smooth hypersurface $S$. 

We call $D$ the basin of attraction of $\bar{x}_1$ and formulate the following set of assumptions.
\begin{itemize}
  \item[(C1)] All the trajectories of the deterministic system $\dot{x}=b(x)$ started on $S$ remain in $S$ and converge to a single equilibrium point $x^* \in S$; besides, the matrix $\nabla b(x^*)$ possesses $d-1$ eigenvalues with negative real part and a single positive eigenvalue $\lambda^*$.
  \item[(C2)] Denoting by $V$ the quasipotential with respect to $\bar{x}_1$ still defined by~\eqref{eq:V}, there exists a unique (up to time shift) trajectory $\rho = (\rho_t)_{t \in \R} \subset D$ such that
  \begin{equation}
    \lim_{t \to -\infty} \rho_t = \bar{x}_1, \qquad \lim_{t \to +\infty} \rho_t = x^*, \qquad \text{and} \quad V(x^*) = \mathcal{A}_{-\infty,+\infty}[\rho].
  \end{equation} 
  \item[(C3)] $V$ is smooth in the neighborhood of $(\rho_t)_{t \in \R}$, and the vector field $\ell$ defined by the identity $b=-\nabla V+\ell$ satisfies the orthogonality relation~\eqref{eq:transv}.
\end{itemize}
In this context, $V$ reaches its minimum on $S$ at the point $x^*$, and we refer to~\cite{Ber13} for a proof of the Arrhenius law~\eqref{eq:Arrh} based on the Freidlin--Wentzell theory. Furthermore, the path $\rho$ is called the \emph{instanton} and it satisfies
\begin{equation}
  \forall t \in \R, \qquad \dot{\rho}_t = \nabla V(\rho_t) + \ell(\rho_t).
\end{equation}
As a consequence, for any $t \in \R$, the fluctuation path $(\varphi^x_s)_{s \leq 0}$ joining $\bar{x}$ to $x=\rho_t$ coincides with the instanton, in the sense that 
\begin{equation}\label{eq:varphirho}
  \forall s \leq 0, \qquad \varphi^x_s = \rho_{s+t}.
\end{equation}

In order to describe the prefactor $L^\epsilon_D$ in this case, we formulate the following supplementary assumption.
\begin{itemize}
  \item[(C4)] The matrix $H^* = \lim_{t \to +\infty} \nabla^2 V(\rho_t)$ exists and has $d-1$ positive eigenvalues and $1$ negative eigenvalue.
\end{itemize}

Under Assumptions~(C1--4), a formula was obtained in~\cite[Eq.~(1.10)]{BouRey16} to describe the sharp asymptotics of the expected time taken by the process to reach the neighborhood of $\bar{x}_2$, which is the contents of the so-called \emph{Eyring--Kramers formula} in the context of reversible diffusion processes~\cite{BovEckGayKle04}. Large deviation theory shows that with overwhelming probability, the path taken by the process to reach $\bar{x}_2$ passes close to $x^*$. At this point, since $b(x^*)=0$, the process has a probability close to $1/2$ to turn back into $D$ and a probability close to $1/2$ to reach the neighbourood of $\bar{x}_2$. Therefore, the expected time taken by the process to exit $D$ is half the time described by the Eyring--Kramers formula. As a consequence, dividing the right-hand side of~\cite[Eq.~(1.10)]{BouRey16} by $2$, we get the estimate
\begin{equation}\label{eq:LDek}
  L^\epsilon_D \mysimeq_{\epsilon\dto 0} \frac{\pi}{\lambda^*}\sqrt{\frac{|\det H^*|}{\det \nabla^2 V(\bar{x}_1)}}\exp\left(\int_{t=-\infty}^{+\infty} \dive\ell(\rho_t)\dd t\right)
\end{equation}
for the prefactor to the mean exit time from $D$.

\begin{remark}
  Notice that in the case addressed in Subsection~\ref{ss:nc}, $L^\epsilon_D$ is proportional to $\sqrt{\epsilon}$, while in the present case, $L^\epsilon_D$ does not depend on $\epsilon$.
\end{remark}

\subsection{Comments on the nonlocal contribution} In the formula obtained for $L^\epsilon_D$ in the cases of both Subsections~\ref{ss:nc} and~\ref{ss:ek}, the \emph{nonlocal contribution} discussed at the end of Section~\ref{s:C} appears. In particular for the Eyring--Kramers formula, the identity~\eqref{eq:varphirho} allows to relate the integral term of~\eqref{eq:LDek} to fluctuation paths by the remark that
\begin{equation}
  \int_{t=-\infty}^{+\infty} \dive\ell(\rho_t)\dd t = \lim_{t \to +\infty} \int_{s=-\infty}^0 \dive\ell(\varphi^{\rho_t}_s)\dd s.
\end{equation}
As a consequence, in the sequel of this paper, we focus on the numerical evaluation of this term.

Recent rigorous derivations of the Eyring--Kramers formula for nonreversible diffusion processes have been obtained in~\cite{LanMarSeo19,LeeSeo}, but to our knowledge, they are restricted to the case of processes whose invariant distribution is given by the  measure~\eqref{eq:Gibbs} and therefore do not include the nonlocal contribution.


\section{Effective computation of the prefactor $C^\epsilon$}\label{s:num}

In Section~\ref{s:exit}, we showed that sharp asymptotics for mean exit times essentially follow from an accurate computation of the prefactor $C^\epsilon$ to the stationary distribution, which was the object of Sections~\ref{s:C} and~\ref{s:path}. In this section, we therefore come back to the framework of Assumptions~(A1--4) and focus on the numerical evaluation of this prefactor, whose equivalent is given by~\eqref{eq:equivC}. 

For a given $x \in \R^d$, the evaluation of the right-hand side of~\eqref{eq:equivC} requires the computation of the following quantities:
\begin{itemize}
  \item the fluctuation path $(\varphi^x_s)_{s \leq 0}$,
  \item the divergence of $\ell$ along the fluctuation path,
  \item the determinant of $\nabla^2 V(\bar{x})$.
\end{itemize}
The main difficulty to access these quantities is that in general, the transverse decomposition~\eqref{eq:decomp} of the vector field $b$ is not explicit, and therefore neither $\nabla V$ nor $\ell$ are straightforward to obtain. 

As a first step, one can remark that thanks to~\eqref{eq:decomp},
\begin{equation}\label{eq:divdelta}
  \forall s \leq 0, \qquad \dive \ell(\varphi^x_s) = \dive b(\varphi^x_s) + \Delta V(\varphi^x_s),
\end{equation}
so that computing the Hessian matrix $\nabla^2 V(\varphi^x_s)$ of the quasipotential along the fluctuation path is sufficient to obtain $\dive \ell(\varphi^x_s)$ and $\det \nabla^2 V(\bar{x})$. 

Motivated by the large deviation principle~\eqref{eq:FW}, specific methods have been developed in the computational physics community to evaluate the quasipotential $V$ (we provide more context in Remark~\ref{rk:schemes} below). The \emph{geometric Minimum Action Method} (gMAM)~\cite{VanHey08:JCP,HeyVan08:CPAM} is a numerical procedure which computes the fluctuation path $\varphi^x$ and returns the value of $V(x)$ given by~\eqref{eq:VMAP}. Once we are supplied with $\varphi^x$, we could virtually iterate the method described above in order to compute the values of $V$ in the neighborhood of the fluctuation path, and hence approximate $\nabla^2 V(\varphi^x_s)$ for selected points $s$ on a time grid. However, computing a fluctuation path for each evaluation of $V$ is too costly and we shall look for an algorithm that avoids doing so. Our method relies on the fact that $\nabla^2 V(\varphi_t^x)$ satisfies a forward matrix Riccati equation, which was already observed in~\cite{Lud75,MaiSte97}.

\begin{proposition}\label{prop:RiccH}
  Let $x \in \R^d$. Under Assumption~(A3), the family of matrices $(H^x_t)_{t \leq 0}$ defined by
  \begin{equation}
    H^x_t = \nabla^2 V(\varphi_t^x)
  \end{equation}
  satisfies the \emph{forward} matrix Riccati equation
  \begin{equation}\label{eq:RiccH}
    \dot{H}^x_t = -2(H^x_t)^2 + {Q_t^x}^{\top} H^x_t + H^x_tQ_t^x + R_t^x,
  \end{equation}
  complemented with the limit condition
  \begin{equation}\label{eq:RiccH:lim}
    \lim_{t \to -\infty} H^x_t = \nabla^2 V(\bar{x}).
  \end{equation}
\end{proposition}
\begin{proof}
  By~\eqref{eq:varphi}, the time derivative of $H^x_t$ writes
  \begin{equation}
    \dot{H}^x_t = \sum_{k=1}^d (\partial_k V(\varphi^x_t) + \ell_k(\varphi^x_t))\partial_k \nabla^2 V(\varphi^x_t).
  \end{equation}
  On the one hand, using~\eqref{eq:decomp} and~\eqref{eq:QR} yields
  \begin{equation}
    \begin{aligned}
      \sum_{k=1}^d \partial_k V(\varphi^x_t)\partial_k \nabla^2 V(\varphi^x_t) & = \sum_{k=1}^d \partial_k V(\varphi^x_t)\nabla^2 (-b_k(\varphi^x_t)+\ell_k(\varphi^x_t))\\
      & = R_t^x + \sum_{k=1}^d \partial_k V(\varphi^x_t)\nabla^2 \ell_k(\varphi^x_t),
    \end{aligned}
  \end{equation}
  while on the other hand, the identity~\eqref{eq:D2transv} yields
  \begin{equation}
    \begin{aligned}
      & \sum_{k=1}^d \ell_k(\varphi^x_t)\partial_k \nabla^2 V(\varphi^x_t)\\
      & = -\nabla^2 V(\varphi^x_t) \nabla \ell(\varphi^x_t) - \nabla \ell(\varphi^x_t)^{\top} \nabla^2 V(\varphi^x_t) - \sum_{k=1}^d \partial_k V(\varphi^x_t)\nabla^2 \ell_k(\varphi^x_t),
    \end{aligned}
  \end{equation}
  so that
  \begin{equation}
    \dot{H}^x_t = -\nabla^2 V(\varphi^x_t) \nabla \ell(\varphi^x_t) - \nabla \ell(\varphi^x_t)^{\top} \nabla^2 V(\varphi^x_t) + R_t^x.
  \end{equation}
  We now compute
  \begin{equation}
    {Q_t^x}^{\top} H^x_t + H^x_tQ_t^x = 2(H^x_t)^2 - \nabla^2 V(\varphi^x_t) \nabla \ell(\varphi^x_t)  - \nabla \ell(\varphi^x_t)^{\top} \nabla^2 V(\varphi^x_t),
  \end{equation}
  which yields~\eqref{eq:RiccH} and completes the proof.
\end{proof}

\begin{remark}\label{rk:forwback}
  The backward matrix Riccati equation~\eqref{eq:Riccati} and the forward Riccati equation~\eqref{eq:RiccH} are related by the fact that if $H_t$ is a solution to~\eqref{eq:RiccH}, then $K_t = -2H_t$ is a solution to~\eqref{eq:Riccati}. This remark is employed in Appendix~\ref{app:riccati} to solve~\eqref{eq:Riccati} by quadrature.
\end{remark}

As was already noted in Section~\ref{ss:laplace}, the coefficients $Q^x_t$ and $R^x_t$ of the forward matrix Riccati equation~\eqref{eq:RiccH} can be computed from the mere knowledge of $b$ and $\varphi^x$, thanks to the identity~\eqref{eq:V-varphi-b}. As a consequence, $\nabla^2 V(\varphi^x_t)$ can be obtained by numerical integration of~\eqref{eq:RiccH}. Therefore we are left with two tasks: computing the limit condition $\nabla^2 V(\bar{x})$ for~\eqref{eq:RiccH}, and integrating this equation on $(-\infty,0]$. These tasks are discussed in the respective Sections~\ref{ss:barH} and~\ref{ss:integr}. The whole method is then summarized in Section~\ref{ss:sumup}.

\begin{remark}\label{rk:schemes}
  The numerical evaluation of large deviation quantities, such as quasipotentials or prefactors, is known to be a difficult question, even in low-dimensional cases. Indeed, the presence of a small term $\epsilon$ in the second-order part of the infinitesimal generator $\epsilon \Delta + \langle b, \nabla\rangle$ of~\eqref{eq:SDE} may make standard discretization schemes for Fokker--Planck or Kolmogorov equations ill-conditioned and unstable. Therefore, dedicated methods need to be developed. Roughly speaking, two classes of such methods exist: \emph{path-based} methods, such as the gMAM employed in this section, which rely on the computation of minimum action paths, and \emph{mesh-based} solvers which compute the value of the quasipotential $V$ on a predetermined grid. We refer for instance to the recent work by Paskal and Cameron~\cite{PasCam} for an example of the latter class of methods, which has the advantage to also provide an approximation of $\nabla V$ on the grid, but crucially suffers from the curse of dimensionality.
\end{remark}


\subsection{Determination of the limit condition}\label{ss:barH} The $t \to -\infty$ limit condition for $H^x_t$ is given by 
\begin{equation}
  \bar{H} = \nabla^2 V(\bar{x}),
\end{equation}
which we assume to be positive-definite. This matrix satisfies the stationary version of~\eqref{eq:RiccH}, which writes
\begin{equation}\label{eq:RiccHstat}
  2\bar{H}^2 = \bar{Q}^\top \bar{H} + \bar{H}\bar{Q},
\end{equation}
with $\bar{Q} = -\nabla b(\bar{x})$. In the context of optimal control, this equation is referred to as a \emph{continuous time algebraic Riccati equation (CARE)}. Alternatively, since $\bar{x}$ is a stable equilibrium point of the dynamical system $\dot{x}=b(x)$, the eigenvalues of $\bar{Q}$ have nonnegative real parts. Let us assume that these eigenvalues have positive real parts. Then $\bar{H}^{-1}$ solves the continuous Lyapunov equation
\begin{equation}\label{eq:bHLyap}
  2I_d = \bar{H}^{-1}\bar{Q}^\top + \bar{Q}\bar{H}^{-1},
\end{equation}
so that
\begin{equation}\label{eq:bHm1}
  \bar{H}^{-1} = 2 \int_{t=0}^{+\infty} \exp(-t\bar{Q})\exp(-t\bar{Q}^\top)\dd t.
\end{equation}

\subsection{Reparametrization and integration of the system}\label{ss:integr}

The system~(\ref{eq:RiccH}-\ref{eq:RiccH:lim}) is defined on the time interval $(-\infty,0]$, with a limit condition in $t=-\infty$. To facilitate its numerical integration, we first introduce a reparametrization of time by the length of the fluctuation path, in the spirit of gMAM~\cite{HeyVan08:CPAM,VanHey08:JCP}.

The length of the fluctuation path is defined by
\begin{equation}
  L^x = \int_{s=-\infty}^0 \|\dot{\varphi}^x_s\| \dd s = \int_{s=-\infty}^0 \|b(\varphi^x_s)\| \dd s,
\end{equation}
where the second identity follows from the orthogonality relation~\eqref{eq:transv}. For all $t \leq 0$, we denote
\begin{equation}
  \sigma_t = \int_{s=-\infty}^t \|\dot{\varphi}^x_s\| \dd s = \int_{s=-\infty}^t \|b(\varphi^x_s)\| \dd s \in (0,L^x].
\end{equation}
The reparametrized fluctuation path is the trajectory $(\tilde{\varphi}^x_{\sigma})_{\sigma \in [0,L^x]}$ defined by the identity
\begin{equation}
  \forall t \leq 0, \qquad \varphi^x_t = \tilde{\varphi}^x_{\sigma_t},
\end{equation}
and the continuous extension $\tilde{\varphi}^x_0 = \bar{x}$. It is easily observed that, for all $\sigma \in (0,L^x]$,
\begin{equation}\label{eq:ddsigma}
  \frac{\dd}{\dd\sigma} \tilde{\varphi}^x_{\sigma} = \frac{\nabla V(\tilde{\varphi}^x_{\sigma}) + \ell(\tilde{\varphi}^x_{\sigma})}{\|b(\tilde{\varphi}^x_{\sigma})\|}.
\end{equation}

We now define 
\begin{equation}
  \tilde{H}^x_\sigma = \nabla^2 V(\tilde{\varphi}^x_{\sigma}),
\end{equation}
so that $\tilde{H}^x_{\sigma_t} = H^x_t$ for all $t \leq 0$. Then the family $(\tilde{H}^x_\sigma)_{\sigma \in [0,L^x]}$ satisfies the problem
\begin{equation}\label{eq:rH}
  \left\{\begin{aligned}
    \frac{\dd}{\dd \sigma} \tilde{H}^x_\sigma &= \|b(\tilde{\varphi}^x_{\sigma})\|^{-1}\left(-2(\tilde{H}^x_\sigma)^2 + (\tilde{Q}_\sigma^x)^\top \tilde{H}^x_\sigma + \tilde{H}^x_\sigma \tilde{Q}_\sigma^x + \tilde{R}_\sigma^x\right), \qquad \sigma \in (0,L^x],\\
    \tilde{H}^x_0 &= \bar{H},
  \end{aligned}\right.
\end{equation}
with
\begin{equation}
  \tilde{Q}_{\sigma_t}^x = Q^x_t, \qquad \tilde{R}_{\sigma_t}^x = R^x_t.
\end{equation}

Notice that when $\sigma \dto 0$, both $\|b(\tilde{\varphi}^x_{\sigma})\|$ and $-2(\tilde{H}^x_\sigma)^2 + (\tilde{Q}_\sigma^x)^\top \tilde{H}^x_\sigma + \tilde{H}^x_\sigma \tilde{Q}_\sigma^x + \tilde{R}_\sigma^x$ vanish. Therefore in order to integrate this system starting from $\sigma=0$, we have to provide an \emph{a priori} estimate of the $\sigma \dto 0$ limit of
\begin{equation}\label{eq:ddsigmH}
  \frac{\dd}{\dd \sigma} \tilde{H}^x_\sigma = \sum_{l=1}^d \frac{\dd}{\dd \sigma} \tilde{\varphi}^x_{l,\sigma} \partial_l \nabla^2 V(\tilde{\varphi}^x_{\sigma}).
\end{equation}

First, a first-order expansion in~\eqref{eq:ddsigma} yields, in the $\sigma \dto 0$ regime,
\begin{equation}\label{eq:ddsigmat}
  \frac{\dd}{\dd \sigma} \tilde{\varphi}^x_\sigma \simeq \frac{\left(\nabla^2 V(\bar{x}) + \nabla \ell(\bar{x})\right)\left(\tilde{\varphi}^x_\sigma-\bar{x}\right)}{\left\|\nabla b(\bar{x})\left(\tilde{\varphi}^x_\sigma-\bar{x}\right)\right\|} = \frac{\left(2\bar{H} + \nabla b(\bar{x})\right)\left(\tilde{\varphi}^x_\sigma-\bar{x}\right)}{\left\|\nabla b(\bar{x})\left(\tilde{\varphi}^x_\sigma-\bar{x}\right)\right\|}
\end{equation}
thanks to~\eqref{eq:decomp}. Each individual term in the right-hand side is computable, so that one can evaluate the quantities ${\frac{\dd}{\dd \sigma} \tilde{\varphi}^x_{l,\sigma}}_{|\sigma=0}$ for all $l \in \{1, \ldots, d\}$.

It remains to compute the third derivatives of $V$ at $\bar{x}$ in order to evaluate the matrices $\partial_l \nabla^2 V(\bar{x})$, $l \in \{1, \ldots, d\}$. To this aim, we take the derivative of~\eqref{eq:D2transv} with respect to the $l$-th coordinate, and evaluate the result at $\bar{x}$. Using the fact that $\ell(\bar{x})$ and $\nabla V(\bar{x})$ vanish, we get the matrix identity
\begin{equation}
  \begin{aligned}
    & 0 = \sum_{k=1}^d \left(\partial_k \nabla^2 V(\bar{x})\right) \left(\partial_l \ell_k(\bar{x})\right) + \left(\partial_l \nabla^2 V(\bar{x})\right)\left(\nabla \ell(\bar{x})\right) + \left(\nabla^2 V(\bar{x})\right) \left(\partial_l \nabla \ell(\bar{x})\right)\\
    & + \left(\partial_l \nabla \ell(\bar{x})\right)^\top \left(\nabla^2 V(\bar{x})\right) + \left(\nabla \ell(\bar{x})\right)^\top \left(\partial_l \nabla^2 V(\bar{x})\right) + \sum_{k=1}^d \left(\partial_{kl}V(\bar{x})\right)\left(\nabla^2 \ell_k(\bar{x})\right).
  \end{aligned}
\end{equation}
Substituting the derivatives of $\ell$ with those of $b+\nabla V$, and introducing the notations
\begin{equation}
  v_{ijl} = \partial_{ijl} V(\bar{x}), \quad h_{ij} = \partial_{ij} V(\bar{x}), \quad \beta_{i,j} = \partial_j b_i(\bar{x}), \quad \gamma_{i,jl} = \partial_{jl} b_i(\bar{x}),
\end{equation}
we finally obtain that, for all $i,j,l \in \{1, \ldots, d\}$,
\begin{equation}\label{eq:vijk}
  \begin{aligned}
    & 0 = \sum_{k=1}^d v_{ijk}(\beta_{k,l}+h_{kl}) + v_{ikl}(\beta_{k,j}+h_{jk}) + h_{ik}(\gamma_{k,jl}+v_{jkl})\\
    & \qquad + (\gamma_{k,il}+v_{ikl})h_{jk} + (\beta_{k,i}+h_{ik})v_{jkl} + h_{kl}(\gamma_{k,ij}+v_{ijk}).
  \end{aligned}
\end{equation}
Since the coefficients $h_{ij}$, $\beta_{i,j}$ and $\gamma_{i,jl}$ are known, the system of equations above induces $d^3$ linear relations between the $d^3$ unknown coefficients $v_{ijl}$ --- more precisely, since both the left-hand side of~\eqref{eq:vijk} and the value of $v_{ijl}$ are invariant by permutation of the indices $i$, $j$ and $l$, the number of independent linear relations and unknown coefficients is reduced to $d(d+1)(d+2)/6$. The resolution of this system allows to reconstruct the matrices $\partial_l \nabla^2 V(\bar{x})$ and completes the computation of~\eqref{eq:ddsigmH}. Notice that these matrices also possess an explicit formulation as an integral along the fluctuation path associated with the linearized stochastic differential equation
\begin{equation}
  \dd \tilde{X}^{\epsilon}_s = \nabla b(\bar{x})\left(\tilde{X}^{\epsilon}_s-\bar{x}\right)\dd s + \sqrt{2\epsilon} \dd W_s, \qquad s \geq 0,
\end{equation}
see~\cite[Section~3.3 and Equation~(3.39)]{BouNarGaw16}.


\subsection{Conclusion}\label{ss:sumup}

Given $x \in \R^d$ and $\epsilon > 0$, the numerical procedure sketched above to compute the right-hand side of~\eqref{eq:equivC} can be summarized in the following steps.
\begin{enumerate}
  \item Compute the fluctuation path $\varphi^x$, for example using gMAM~\cite{HeyVan08:CPAM,VanHey08:JCP}. From this step, the value of $V(x)$ along the fluctuation path can be deduced thanks to~\eqref{eq:VMAP}, which allows to evaluate the right-hand side of~\eqref{eq:FW}.
  \item Solve the stationary matrix Riccati equation~\eqref{eq:RiccHstat} to get $\bar{H} = \nabla^2 V(\bar{x})$. This can be done either:
  \begin{itemize}
    \item by computing the integral~\eqref{eq:bHm1}, which solves the Lyapunov equation~\eqref{eq:bHLyap}, and inverting the result;
    \item or by using a numerical solver for the CARE~\eqref{eq:RiccHstat} directly.
  \end{itemize}
  \item Compute $L^x$ and for a given number of time steps $N \gg 1$, compute times $0 = t_N > t_{N-1} > \cdots > t_1 > t_0 = -\infty$ such that
  \begin{equation}
    \forall n \in \{1, \ldots, N-1\}, \qquad \int_{s=t_n}^{t_{n+1}} \|b(\varphi^x_s)\|\dd s = \theta,
  \end{equation}
  with $\theta = L^x/N$.
  \item Using~\eqref{eq:ddsigmat}, compute 
  \begin{equation}
    {\frac{\dd}{\dd\sigma} \tilde{\varphi}^x_\sigma}_{|\sigma=0} \simeq  \frac{\left(2\bar{H} + \nabla b(\bar{x})\right)\left(\varphi^x_{t_1}-\bar{x}\right)}{\left\|\nabla b(\bar{x})\left(\varphi^x_{t_1}-\bar{x}\right)\right\|},
  \end{equation}
  and solve~\eqref{eq:vijk} to get the matrices $\partial_l \nabla^2 V(\bar{x})$, $l \in \{1, \ldots, d\}$.
  \item Compute the approximation $\tilde{H}^{[n]}$ of $\tilde{H}^x_{n\theta} = H_{t_n}$ by integrating the matrix-valued differential equation
  \begin{equation}\label{eq:odetointegrate}
    \frac{\dd}{\dd \sigma} \tilde{H}^x_\sigma = \frac{-2(\tilde{H}^x_\sigma)^2 + (\tilde{Q}_\sigma^x)^\top \tilde{H}^x_\sigma + \tilde{H}^x_\sigma \tilde{Q}_\sigma^x + \tilde{R}_\sigma^x}{\|b(\tilde{\varphi}^x_{\sigma})\|}
  \end{equation}
  on the grid $\sigma \in \{\theta, 2\theta, \ldots, N\theta\}$, with initial conditions
  \begin{equation}
    \left\{\begin{aligned}
      & \tilde{H}^{[0]} = \bar{H},\\
      & \tilde{H}^{[1]} = \tilde{H}^{[0]} + \theta \sum_{l=1}^d {\frac{\dd}{\dd\sigma} \tilde{\varphi}^x_{l,\sigma}}_{|\sigma=0}\partial_l \nabla^2 V(\bar{x}).
    \end{aligned}\right.
  \end{equation}
  Many schemes can be employed for this numerical integration; in the example of Section~\ref{s:na}, we use a first-order implicit Euler scheme, which is observed to have satisfying stability properties.
  \item Deduce from~\eqref{eq:divdelta} that
  \begin{equation}
    \int_{s=-\infty}^0 \dive \ell(\varphi^x_s)\dd s \simeq \sum_{n=0}^{N-1} (t_{n+1}-t_n)\left(\dive b(\varphi^x_{t_n}) + \tr \tilde{H}^{[n]}\right),
  \end{equation}
  where we recall that the times $t_0, \ldots, t_N$ are chosen so that the length of the instanton on each time interval $(t_n,t_{n+1})$ be equal to $\theta$.
\end{enumerate}


\section{Numerical illustration}\label{s:na}

In this section, we apply the method devised in Section~\ref{s:num} to a two-dimensional process which exhibits \emph{bistability}, in the sense that the associated vector field $b : \R^2 \to \R^2$ possesses two stable equilibrium points. We are therefore in the context of Subsection~\ref{ss:ek} and our purpose is to numerically approximate the various quantities appearing in the right-hand side of~\eqref{eq:LDek}.

We shall fix one equilibrium point $\bar{x}$ and first compute the instanton $(\rho_t)_{t \in \R}$. By~\eqref{eq:varphirho}, the integral term involved in the computation of the prefactor $C^\epsilon(x)$ to the stationary distribution at the point $x$ writes
\begin{equation}\label{eq:J-instanton}
  \int_{s=-\infty}^0 \dive \ell(\varphi^x_s)\dd s = \int_{s=-\infty}^t \dive \ell(\rho_s)\dd s.
\end{equation}
For notational convenience, we denote this quantity by $J(t)$. Hence, computing $J(t)$ for any $t \in \R$ amounts to computing a whole family of prefactors $C^\epsilon$, at points $x=\rho_t$. In the $t \to +\infty$ limit, we shall finally obtain the value of the prefactor $L^\epsilon_D$ to the mean exit time from the basin of attraction of $\bar{x}_1$, as is described in Subsection~\ref{ss:ek}.

In the present section, we follow the steps of the procedure detailed in Subsection~\ref{ss:sumup}. Step~1, which corresponds to the computation of the fluctuation path, is addressed in Subsection~\ref{ss:6.1}, where we first present the example. The computation of $\bar{H}$ (Step~2) is performed in Subsection~\ref{ss:6.2}. Anticipating on Step~4, we compute the third derivatives of $V$ at $\bar{x}$ in Subsection~\ref{ss:6.3}. Steps~3, 4 and~5 are then addressed in Subsection~\ref{ss:6.4}, which yields $\tilde{H}^{[n]}$. Since we are considering fluctuation paths which finally reach the saddle-point (that is to say, the instanton), the reparametrization by the arclength of the instanton display a singularity when approaching this point, so that the computation of $\tilde{H}^{[n]}$ becomes unstable. This point is treated in Subsection~\ref{ss:6.5}. Last, the overall value of $J(t)$ is computed in Subsection~\ref{ss:6.6}, which corresponds to Step~6 of the procedure.

\subsection{Presentation of the example}\label{ss:6.1}

For $\alpha > 0$, we consider the potential function on $\R^2$ defined by
\begin{equation}
  V(x_1,x_2) = v_1(x_1) + v_2(x_2), \qquad v_1(x_1) = \frac{x_1^4}{4} - \frac{x_1^2}{2}, \quad v_2(x_2) = \alpha \frac{x_2^2}{2},
\end{equation}
the critical points of which are $(-1,0)$, $(0,0)$ and $(1,0)$. It is easily checked that the first and third points are stable equilibria of the dynamical system $\dot{x}=-\nabla V(x)$, while the second point is stable in the direction of $x_2$ but unstable in the direction of $x_1$.

For a smooth scalar field $c : \R^2 \to \R$ to be chosen below, let us define the vector field $b$ by
\begin{equation}
  b = -\nabla V + \ell, \qquad \ell(x_1,x_2) = c(x_1,x_2)\left(\begin{matrix}
    -v'_2(x_2)\\
    v'_1(x_1)
  \end{matrix}\right).
\end{equation}
For any choice of $c$, the Hamilton--Jacobi equation~\eqref{eq:HJ} is satisfied, the vector field $b$ vanishes at the same points as $-\nabla V$, and the points $(-1,0)$ and $(1,0)$ are stable equilibria of the dynamical system $\dot{x}=b(x)$. The \emph{saddle-point} $(0,0)$ is stable in one direction and unstable in one direction, but these directions may differ from the canonical vectors of $\R^2$.

Let us denote $\bar{x} = (-1,0)$ and $x^* = (0,0)$. The \emph{instanton} $(\rho_t)_{t \in \R}$ is the heteroclinic orbit of the dynamical system
\begin{equation}
  \dot{x} = \nabla V(x) + \ell(x)
\end{equation}
joining $\bar{x}$ in $t=-\infty$ to $x^*$ in $t=+\infty$. The instanton, the level lines of $V$ and the field lines of $b$ are plotted on Figure~\ref{fig:lines} for the choice 
\begin{equation}\label{eq:cexpl}
  c(x_1,x_2) = \beta x_1, \qquad \beta > 0.
\end{equation}
The numerical illustrations of this section are plotted for $\alpha=0.5$ and $\beta=3$.

\begin{figure}[htbp]
  \centering
  \includegraphics[width=\textwidth]{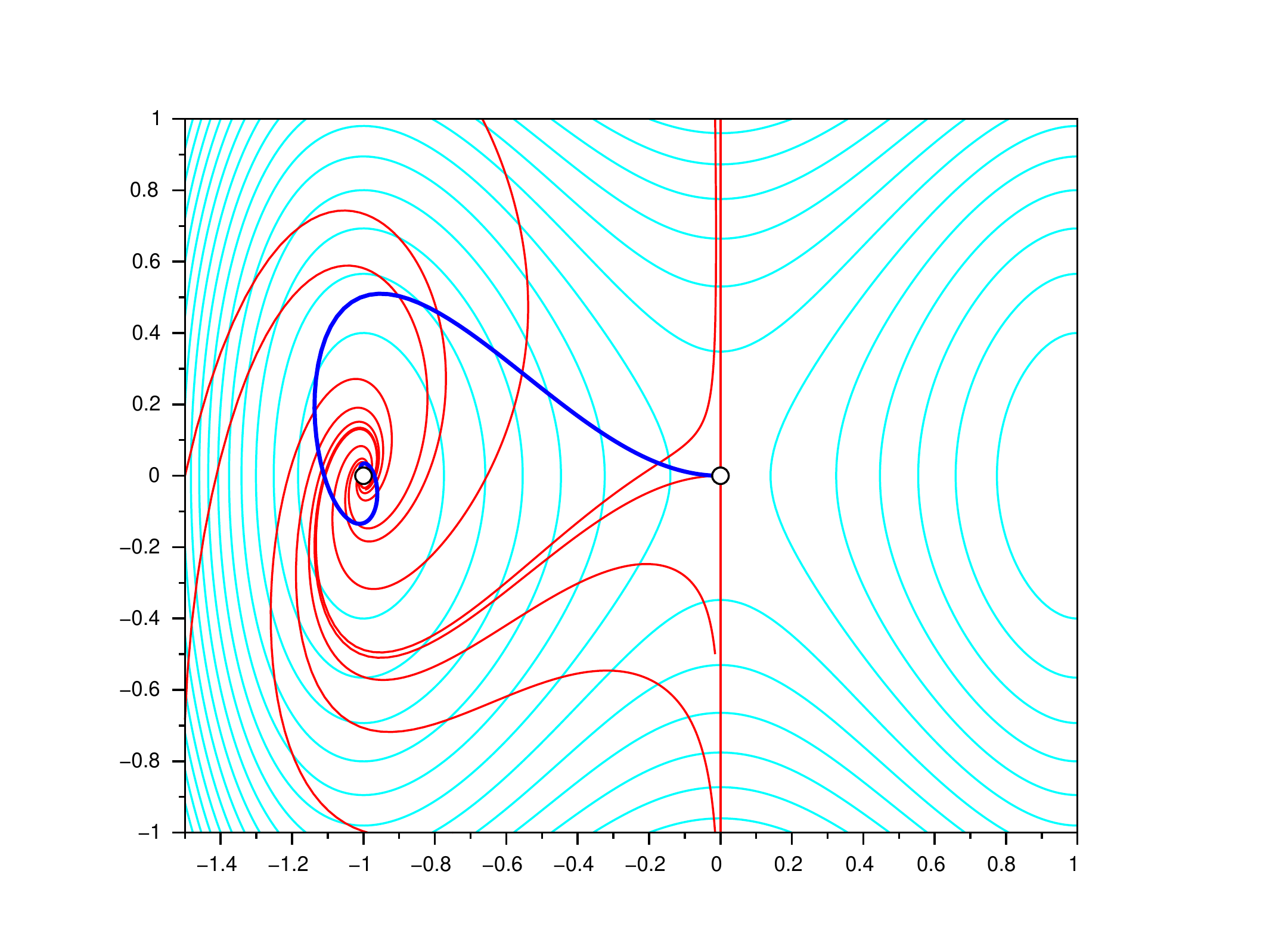}
  \caption{The level lines of $V$ (pale blue), some field lines associated with $b$ (red) and the instanton (thick blue) joining $\bar{x}=(-1,0)$ to $x^*=(0,0)$.}
  \label{fig:lines}
\end{figure}

\subsection{Stationary matrix Riccati equation}\label{ss:6.2} With our definition of the vector field $b$ and the choice~\eqref{eq:cexpl} for the scalar field $c$, the matrix $\bar{Q} = -\nabla b(\bar{x})$ appearing in~\eqref{eq:RiccHstat} writes
\begin{equation}
  \bar{Q} = \left(\begin{matrix}
    v''_1(\bar{x}_1) & c(\bar{x})v''_2(\bar{x}_2)\\
    -c(\bar{x})v''_1(\bar{x}_1) & v''_2(\bar{x}_2)
    \end{matrix}\right) = \left(\begin{matrix}
    2 & -\beta\alpha\\
    2\beta & \alpha
    \end{matrix}\right),
\end{equation}
where we have used the fact that $\nabla V(\bar{x})=0$ for the first identity. The numerical resolution of~\eqref{eq:RiccHstat} yields the expected result
\begin{equation}
  \bar{H} = \left(\begin{matrix}
    2 & 0\\
    0 & \alpha
    \end{matrix}\right).
\end{equation}

\subsection{Third derivatives of $V$ at $\bar{x}$}\label{ss:6.3}

In order to determine the initial condition for the computation of $\tilde{H}^{[n]}$, we write the coefficients appearing in~\eqref{eq:vijk}:
\begin{itemize}
  \item the computation of $\bar{H}$ performed above yields 
  \begin{equation}
    h_{11}=2, \quad h_{12}=0, \quad h_{22}=\alpha;
  \end{equation}
  \item the computation of $\nabla b(\bar{x})$ yields 
  \begin{equation}
    \beta_{1,1}=-2, \quad \beta_{1,2}=\beta\alpha, \quad \beta_{2,1}=-\beta, \quad \beta_{2,2}=-\alpha;
  \end{equation}
  \item the computation of the second derivatives of $b$ yields
  \begin{equation}
    \left\{\begin{aligned}
      &\gamma_{1,11} = 6, \quad \gamma_{1,12} = -\beta\alpha, \quad \gamma_{1,22} = 0,\\
      &\gamma_{2,11} = 10 \beta, \quad \gamma_{2,12} = 0, \quad \gamma_{2,22} = 0.
    \end{aligned}\right.
  \end{equation}
\end{itemize}
The system of linear equations~\eqref{eq:vijk} contains $4$ equations, corresponding to the choices $\{i,j,l\} = \{1,1,1\}, \{1,1,2\}, \{1,2,2\}, \{2,2,2\}$, which write
\begin{equation}
  \left\{\begin{aligned}
    0 &= 36 + 6v_{111} -3\beta v_{112},\\
    0 &= 6\beta\alpha +\beta\alpha v_{111} + (4+\alpha)v_{112} - 2\beta v_{122},\\
    0 &= 2\beta\alpha v_{112} + 2(1+\alpha)v_{122} - \beta v_{222},\\
    0 &= 3\beta\alpha v_{122} + 3\alpha v_{222}.
  \end{aligned}\right. 
\end{equation}
The unique solution of this system is
\begin{equation}
  v_{111} = -6, \quad v_{112} = v_{122} = v_{222} = 0,
\end{equation}
therefore we recover
\begin{equation}
  \partial_1 \nabla^2 V(\bar{x}) = \left(\begin{matrix}
    -6 & 0\\
    0 & 0
  \end{matrix}\right), \qquad \partial_2 \nabla^2 V(\bar{x}) = \left(\begin{matrix}
    0 & 0\\
    0 & 0
  \end{matrix}\right),
\end{equation}
as expected from the analytical expression of $V$.

\subsection{Computation of $\tilde{H}^{[n]}$}\label{ss:6.4}

The instanton on Figure~\ref{fig:lines} is computed for times $t$ ranging from $t_\mathrm{min}$ such that $\|\rho_{t_\mathrm{min}}-\bar{x}\| \ll 1$ to $t_\mathrm{max}$ such that $\|\rho_{t_\mathrm{max}}-x^*\| \ll 1$. The length of the instanton is thus
\begin{equation}
  L = \int_{t=-\infty}^{+\infty} \|b(\rho_t)\| \dd t \simeq \int_{t=t_\mathrm{min}}^{t_\mathrm{max}} \|b(\rho_t)\| \dd t.
\end{equation}
For the example which we are studying, $L \simeq 2.15$.

The parametrization of the instanton by its length is defined by
\begin{equation}
  \sigma_t = \int_{s=-\infty}^t \|b(\rho_s)\|\dd s \simeq \int_{s=t_\mathrm{min}}^t \|b(\rho_s)\|\dd s,
\end{equation}
and we denote
\begin{equation}
  \tilde{\rho}_{\sigma_t} = \rho_t, \qquad \tilde{Q}_\sigma = -\nabla b(\tilde{\rho}_\sigma), \qquad \tilde{R}_\sigma = -\sum_{k=1}^d \partial_k V(\tilde{\rho}_\sigma) \nabla^2 b(\tilde{\rho}_\sigma).
\end{equation}

In order to integrate the differential equation
\begin{equation}\label{eq:MREsigma}
  \frac{\dd}{\dd\sigma} \tilde{H}_\sigma = \frac{-2(\tilde{H}_\sigma)^2 + (\tilde{Q}_\sigma)^\top \tilde{H}_\sigma + \tilde{H}_\sigma \tilde{Q}_\sigma + \tilde{R}_\sigma}{\|b(\tilde{\rho}_\sigma)\|},
\end{equation}
we use standard, first-order Euler schemes. The explicit scheme
\begin{equation}
  \frac{\tilde{H}^{[n+1]}-\tilde{H}^{[n]}}{\theta} = \frac{-2(\tilde{H}^{[n]})^2 + (\tilde{Q}_{n\theta})^\top \tilde{H}^{[n]} + \tilde{H}^{[n]} \tilde{Q}_{n\theta} + \tilde{R}_{n\theta}}{\|b(\tilde{\rho}_{n\theta})\|}
\end{equation}
is observed to be unstable. A semi-implicit scheme is presented in~\cite{DubSai00}; in the present case, it is also observed to be unstable. Following ideas introduced in~\cite{DieEir94,DieEir96}, we finally consider the implicit scheme
\begin{equation}
  \frac{\tilde{H}^{[n+1]}-\tilde{H}^{[n]}}{\theta} = \frac{-2(\tilde{H}^{[n+1]})^2 + (\tilde{Q}_{n\theta})^\top \tilde{H}^{[n+1]} + \tilde{H}^{[n+1]} \tilde{Q}_{n\theta} + \tilde{R}_{n\theta}}{\|b(\tilde{\rho}_{n\theta})\|}
\end{equation}
which requires to solve the CARE 
\begin{equation}\label{eq:CARE-implicit}
  2\tilde{\theta}_n(\tilde{H}^{[n+1]})^2 + \left(\frac{I_2}{2}-\tilde{\theta}_n\tilde{Q}_{n\theta}\right)^\top \tilde{H}^{[n+1]} + \tilde{H}^{[n+1]}\left(\frac{I_2}{2}-\tilde{\theta}_n\tilde{Q}_{n\theta}\right) = \tilde{H}^{[n]} + \tilde{\theta}_n\tilde{R}_{n\theta},
\end{equation}
with $\tilde{\theta}_n = \theta/\|b(\tilde{\rho}_{n\theta})\|$, at each step. This scheme is observed to be stable and convergent.

We point out the fact that, to our knowledge, the theoretical results regarding the numerical analysis of the matrix Riccati equation~\eqref{eq:MREsigma}, such as~\cite{DieEir94,DieEir96,DubSai00}, assume that the matrix $\tilde{R}_\sigma$ remains nonnegative, which then ensures the nonnegativity of the solution $\tilde{H}_\sigma$. In our situation, the matrix $\tilde{H}_\sigma$ is clearly \emph{not} nonnegative in general, except at the initial point $\sigma=0$. Therefore, the result of our numerical simulations are purely empirical, and are not backed up by some rigorous stability or convergence result.

Figure~\ref{fig:HessQP} represents the value of the coefficient $\tilde{H}^{[n]}_{11}$ for $n=0, \ldots, N$, for several choices of the mesh size $\theta$. The actual values of the first coefficient of $\nabla^2 V(\rho_{t_n})$, computed from the analytical expression of $V$, are provided as a benchmark. We observe that the convergence is slow in the neighborhood of the saddle-point $x^*$, which is discussed in the next subsection.

\begin{figure}[htbp]
  \centering
  \includegraphics[width=.9\textwidth]{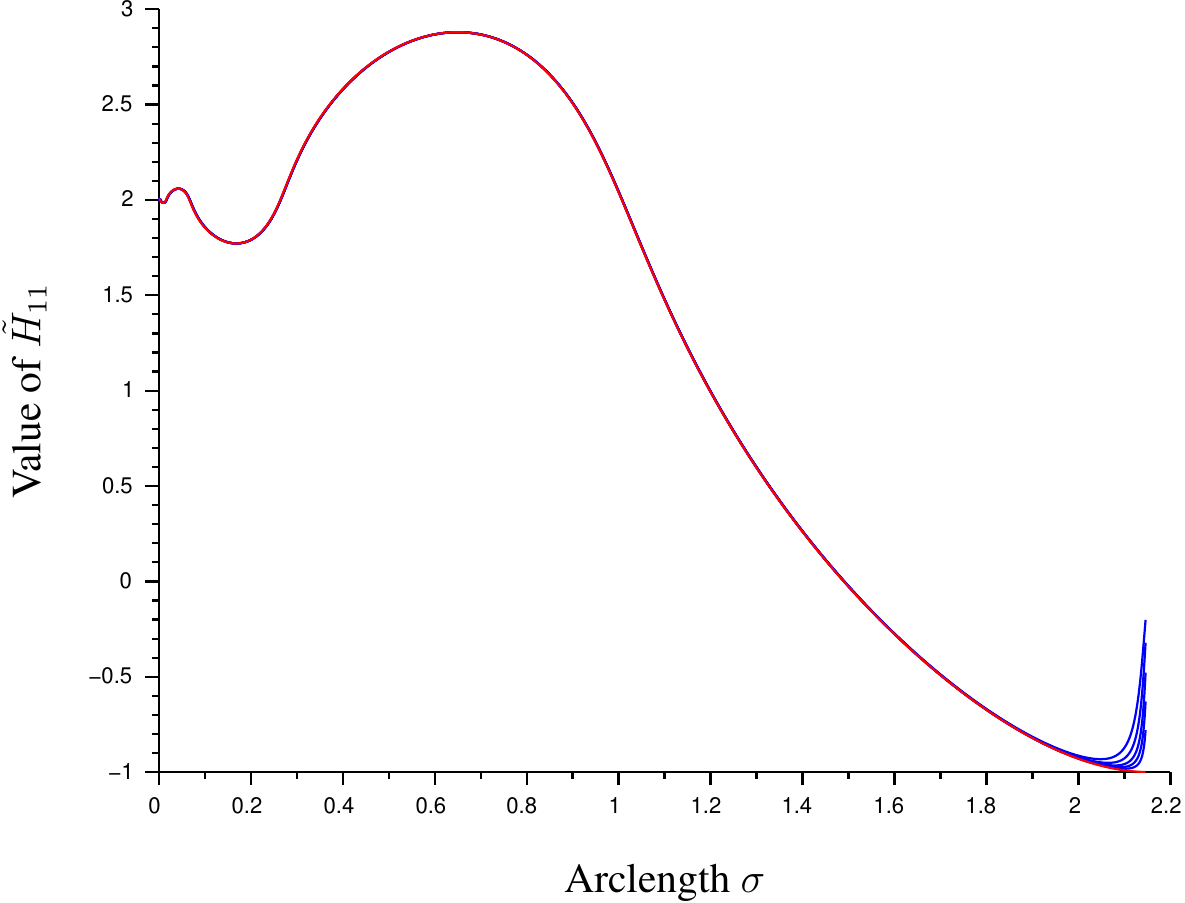}
  \caption{Comparison between the computed value of $\tilde{H}^{[n]}_{11}$ (blue curves) and the actual value of $\partial_{11}V(\rho_{t_n})$ (red), parametrized by the arclength of the instanton $\sigma \in [0,L]$. The smaller the mesh size $\theta$, the better the convergence at the saddle-point. The different choices of $\theta$ correspond to the values $2000$, $4000$, $8000$, $16000$ and $40000$ for the number of steps $N=L/\theta$.}
  \label{fig:HessQP}
\end{figure}

When the dimension $d$ is large, solving the CARE~\eqref{eq:CARE-implicit} at each step of the algorithm may turn out to be costly and lessen the interest of the implicit scheme. In such cases, alternative approaches such as the `fundamental solution method' from~\cite{DieEir94} can be employed. The latter method also allows to implement higher-order schemes.

\subsection{Singularity at the saddle-point}\label{ss:6.5}

When $\sigma$ approaches $L$, the instanton becomes close to the saddle-point $x^*$ and the value of $\|b(\tilde{\rho}_\sigma)\|$ goes to $0$. Therefore the integration of~\eqref{eq:MREsigma} becomes sensitive to the fact that the denominator in the right-hand side takes small values, which causes the singularity observed on Figure~\ref{fig:HessQP}. In order to overcome this issue, we note that, similarly to the initial value $\bar{H} = \tilde{H}_0$, the terminal value
\begin{equation}
  \tilde{H}_L = \lim_{t \to +\infty} H_t = \nabla V(x^*) = H^*
\end{equation}
can be computed by solving the stationary matrix Riccati equation
\begin{equation}
  2H^* = (Q^*)^\top H^* + H^* Q^*, \qquad Q^* = -\nabla b(x^*).
\end{equation}
This remark allows us to implement the following interpolation procedure.
\begin{enumerate}
  \item Fix a threshold $\delta$ such that $\theta \ll \delta \ll L$.
  \item Among the indices $n$ such that $n\theta \geq L-\delta$, select the index $n^*$ for which the linear continuation
  \begin{equation}
    \tilde{H}^{[n]} + \frac{\tilde{H}^{[n+1]}-\tilde{H}^{[n]}}{\theta}(L-n\theta)
  \end{equation}
  is the closest to the terminal value $H^*$, for a given matrix norm.
  \item For $n$ between $n^*$ and the total number of steps $N$, replace the estimation $\tilde{H}^{[n]}$ of $\tilde{H}_{n \theta}$ with the linear interpolation
  \begin{equation}
    \tilde{H}'^{[n]} = \frac{N-n}{N-n^*} \tilde{H}^{[n^*]} + \frac{n-n^*}{N-n^*} H^*.
  \end{equation}
\end{enumerate}

This procedure allows to alleviate the singularity at the saddle-point, as is shown on Figure~\ref{fig:HessQP-saddle}. We observe that the computed value of the first coefficient of $\tilde{H}_\sigma$ is correct up to an error of the order of magnitude $5\%$ (respectively $1\%$) for the coarsest mesh size $\theta = L/2000$ (respectively the finest mesh size $\theta = L/40000$), localized in the area close to the critical point, for distances of order $10\%$ of $L$ or less.

\begin{figure}[htbp]
  \centering
  \includegraphics[width=.9\textwidth]{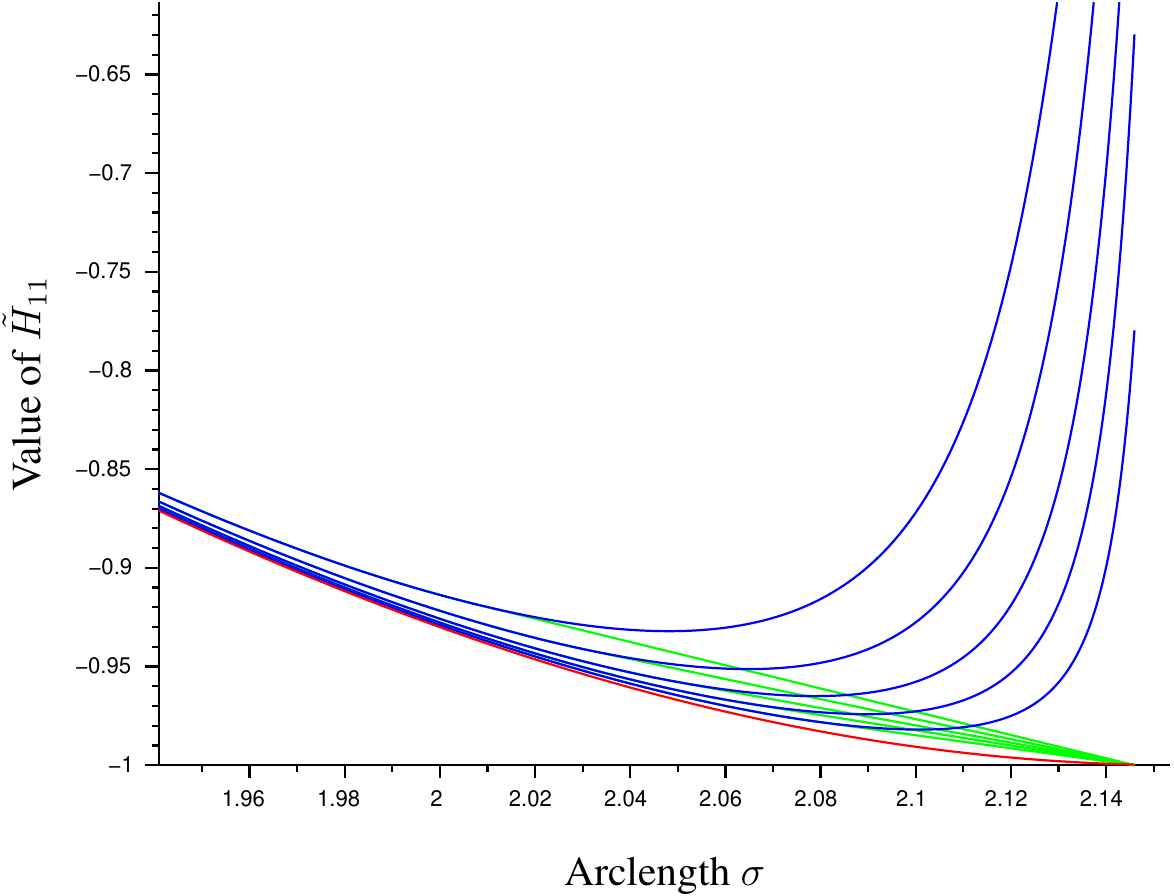}
  \caption{A zoom on the numerical computation of the first coefficient of $\tilde{H}_\sigma$ when $\sigma \uto L$. The red curve is the actual value. The blue curves are the values of $\tilde{H}^{[n]}_{11}$ already shown on Figure~\ref{fig:HessQP}. The green curves are the values of $\tilde{H}'^{[n]}_{11}$ obtained by the interpolation procedure. Here, $\delta=0.2$. }
  \label{fig:HessQP-saddle}
\end{figure}

\subsection{Evaluation of $J(t)$}\label{ss:6.6}

From the quantity $J(t)$, defined by~\eqref{eq:J-instanton}, let us define $\tilde{J}(\sigma)$ for $\sigma \in [0,L]$ by
\begin{equation}
  \tilde{J}(\sigma_t) = J(t).
\end{equation}
The quantity $\tilde{J}(\sigma)$ is the value of the prefactor at the point with arclength $\sigma$ on the instanton. These values, computed from the numerical resolution of the matrix Riccati equation for $(H_t)_{t \in \R}$, are plotted on Figure~\ref{fig:Jlength} (the interpolation procedure at the saddle-point discussed in the previous subsection is employed), for several choices of the mesh size $\theta$. The actual values of $\tilde{J}(\sigma)$, computed from the analytical expression of $V$, are provided as a benchmark. We observe a good agreement, up to an error of the order of magnitude $1\%$ localized in the area close to the critical point, for distances of order $10\%$ of $L$ or less, which supports the efficiency of our method. The evolution of the discretization error on $\tilde{J}(L)$ as a function of $\theta$ is plotted on Figure~\ref{fig:errorJ}, it is observed to be proportional to $\sqrt{\theta}$.

\begin{figure}[htbp]
  \centering
  \includegraphics[width=.9\textwidth]{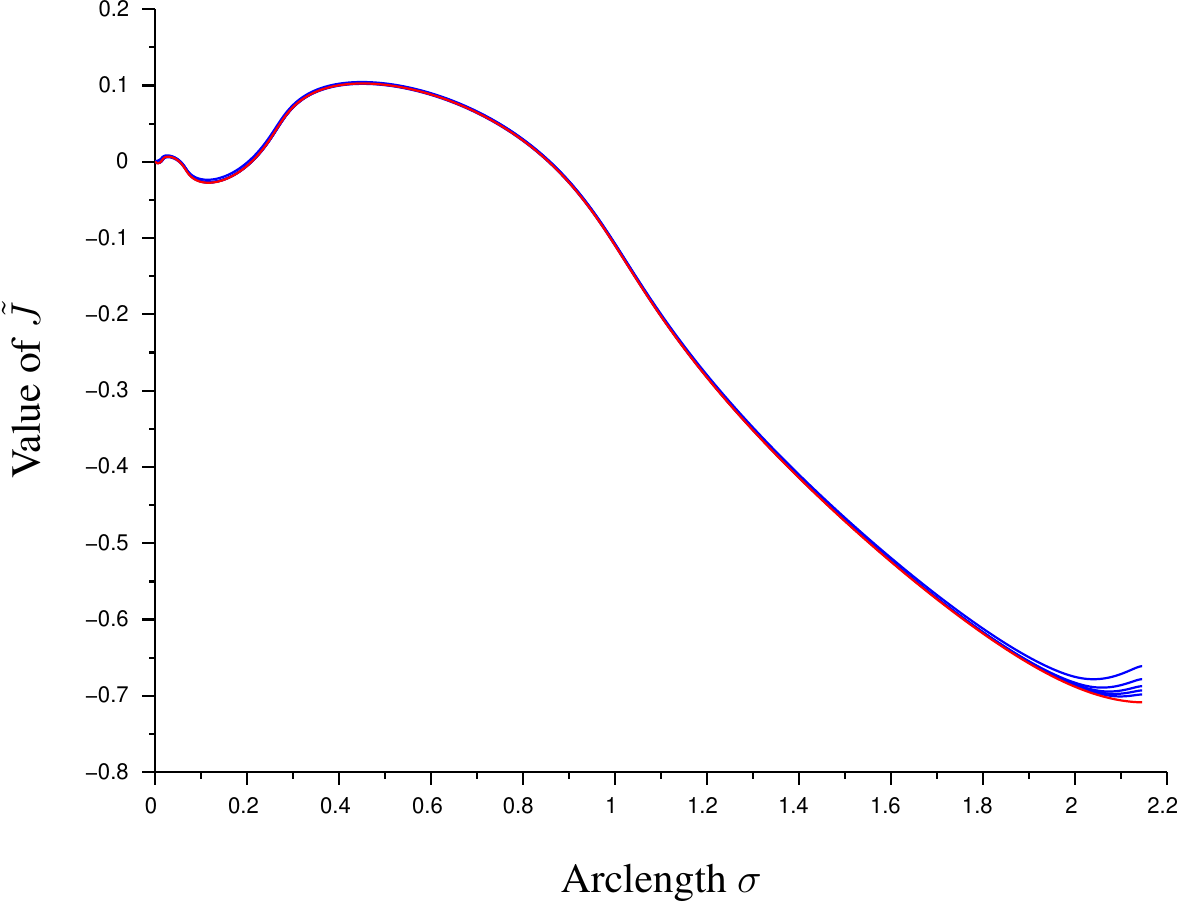}
  \caption{Comparison between the computed value of $\tilde{J}(\sigma)$ (blue curves) and its actual value (red curve), for $\sigma \in [0,L]$. The closer the curve, the smaller $\theta$. The different choices of $\theta$ correspond to the values $2000$, $4000$, $8000$, $16000$ and $40000$ for the number of steps $N=L/\theta$.}
  \label{fig:Jlength}
\end{figure}

\begin{figure}[htbp]
  \centering
  \includegraphics[width=.9\textwidth]{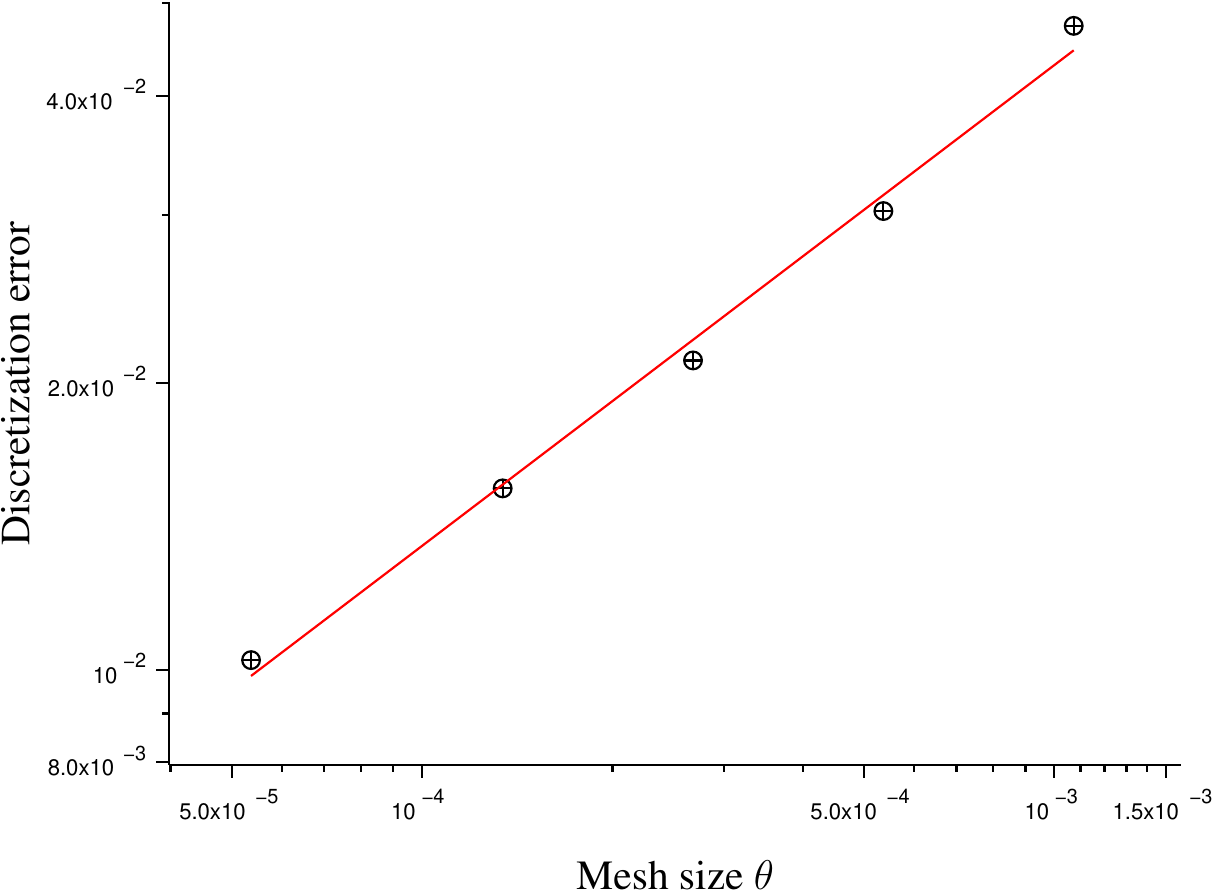}
  \caption{Log-log plot of the discretization error on $\tilde{J}(L)$ as a function of $\theta$, with linear fit of slope $0.50$.}
  \label{fig:errorJ}
\end{figure}


\section{Summary and relation to recent works}\label{s:conc}

In this conclusive section, we summarize the main contributions of the article, on both theoretical and numerical aspects. We also compare these results with other recent works, in particular~\cite{GraSchVan} which was released during the last stages of our work and contains several related ideas and results.

\subsection{Theoretical contributions}

At the conceptual level, the main original contribution of the article is the derivation in Section~\ref{s:path} of a sharp equivalent, when $\epsilon \dto 0$, to the prefactor $C^\epsilon$ of the stationary distribution $P^\epsilon$. It is done in two steps.
\begin{itemize}
  \item[(i)] We perform an asymptotic expansion in the path integral formulation of $P^\epsilon$ (see Eq.~\eqref{eq:Peps:1}) in order to relate the prefactor $C^\epsilon$ with the process of scaled fluctuations $Y^x$ defined by~\eqref{eq:Y}, through the identity~\eqref{eq:u-PI}.
  \item[(ii)] We express $C^\epsilon$ in terms of two quantities $\eta^x_t$ and $K^x_t$, which are related by the matrix Riccati equation~\eqref{eq:Riccati}. We then solve explicitly this equation in the Appendix by a quadrature method, which finally allows to recover a sharp equivalent for $C^\epsilon$.
\end{itemize} 
As is recalled in Section~\ref{s:C}, the expression of the sharp equivalent for $C^\epsilon$ is not new and was already derived by a WKB approximation in~\cite{BouRey16}. Therefore the real novelty here is the sketch of the argument, and in particular the resolution of the matrix Riccati equation.

The fact that large deviation prefactors induced by Gaussian fluctuations around action minimizing paths can be described in terms of solutions to matrix Riccati equations has already been observed in various contexts~\cite{Lud75,MaiSte97,FerGra}. It was put forth in the recent work~\cite{GraSchVan} by Grafke, Schäfer and Vanden-Eijnden, who conducted a thorough derivation of expressions of large deviation prefactors in terms of solutions to matrix Riccati equations, both for finite-time observables and quantities related to stationary distribution. Their work includes many illustrative examples and generalization towards nongaussian noises and infinite-dimensional systems.

Some fundamental ideas in~\cite{GraSchVan} are similar to the present paper; for instance, the use of the Girsanov transform in Proposition~2.2 there makes the scaled fluctuation process $Y^x$ appear in the prefactor $C^\epsilon$ in an equivalent way to our use of the path integral formalism. However, both works differ on several methodological aspects. In particular, both the formulation of matrix Riccati equations and the expression of prefactor asymptotics in~\cite{GraSchVan} involve the fact that the fluctuation path $\varphi^x$ is defined by the forward-backward \emph{Hamiltonian system}
\begin{equation}
  \dot{\varphi}^x_s = b(\varphi^x_s) + 2 \theta^x_s, \qquad \dot{\theta}^x_s = -(\nabla b(\varphi^x_s))^\top \theta^x_s, \qquad s \in [0,t],
\end{equation}
with suitable limit conditions on $\varphi^x_0$ and $\theta^x_t$ depending on the large deviation quantity on which prefactors are computed. In contrast, in the present article, the definition~\eqref{eq:varphi} of fluctuation paths relies on the transverse decomposition~\eqref{eq:decomp} of $b$. As is noted in~\cite[Section~3.2]{GraSchVan}, this is due to the fact that we are merely interested in quantities related with infinite time horizon. Still, this allows us to then derive prefactor asymptotics from the explicit resolution of the matrix Riccati equation~\eqref{eq:Riccati} by quadrature, which as is argued above is the main theoretical contribution of our article and is inherent to the use of the transverse decomposition.

\subsection{Numerical contributions}

The second main contribution of the paper is the formalization in Section~\ref{s:num} of a complete numerical method to compute all quantities involved in the prefactor $C^\epsilon$. In fact, based on the connections recalled in Section~\ref{s:exit} between this prefactor and mean exit times, our numerical procedure enables one to evaluate prefactors to transition times between metastable sets, as is illustrated in Section~\ref{s:na}.

Numerical schemes, based on the resolution of matrix Riccati equations, are also discussed in~\cite{GraSchVan} (but the study of metastable settings is not covered there). A common feature of both works is the preliminary reparametrization of these equations by the arclength of the fluctuation path, which allows to address their long-time behavior. The singularity at $\sigma=0$ induced by this reparametrization, which is addressed in Subsection~\ref{ss:integr}, is also discussed in Remark~3.4 and Appendix~A of~\cite{GraSchVan}, and solved with similar arguments. 

On the other hand, Paskal and Cameron~\cite{PasCam} recently devised a mesh-based method dubbed `Efficient Jet Marcher' which allows to compute the values of $V$ and $\nabla V$ on a grid (see Remark~\ref{rk:schemes}). In Section~6.2 of their article, they applied this method to evaluate the prefactor $L^\epsilon_D$ from~\eqref{eq:LDek} of the mean transition time in a two-dimensional metastable setting. This method, which radically differs from the path-inspired approaches from the present article and~\cite{GraSchVan}, is however currently limited to low-dimensional situations.


\appendix\section{Appendix}

The Appendix is organized as follows. In Section~\ref{app:OE}, the notion of time ordered exponential is introduced and a few properties are stated. Section~\ref{app:riccati} is dedicated to the resolution of the backward matrix Riccati equation appearing in~\eqref{eq:Riccati}. Finally, Section~\ref{app:pfeta} presents the proof of~\eqref{eq:limeta}.

Throughout the Appendix, we work under Assumptions~(A1--4). Furthermore, as is argued in the beginning of Section~\ref{s:path}, our overall purpose here is to emphasize the connections between fluctuation paths, large deviation prefactors, functional Gaussian determinants and matrix Riccati equations. Therefore, we chose not to obscure the exposition of our arguments with the exhaustive mathematical justification of technical details, which can however be checked on a case-by-case basis.


\subsection{Time ordered exponentials}\label{app:OE}

Throughout this section we let $A=(A_t)_{t \leq 0}$ be a bounded family of matrices of size $d \times d$. The choice of $(-\infty,0]$ as the set of times is convenient for our purpose but the contents of this section could easily be adapted to any interval.

For all $t_0, t \leq 0$, we denote by 
\begin{equation}
  M_t = \OrdExp{\int_{s=t_0}^t A_s \dd s} \in \R^{d \times d}
\end{equation}
the solution to the (two-sided) Cauchy problem
\begin{equation}
  \left\{\begin{aligned}
    \dot{M}_t & = A_t M_t, \qquad t \leq 0,\\
    M_{t_0} & = I_d.
  \end{aligned}\right.
\end{equation}
It is called the \emph{time ordered exponential} of $A$ from $t_0$ to $t$. It is related with the notion of \emph{path ordering} which is in particular used in quantum field theory. However, in the context of the present article, we are only dealing with bounded, finite-dimensional matrices, therefore the following properties, which will be used in the sequel, are elementary consequences of the Cauchy--Lipschitz theorem for linear ordinary differential equations.
\begin{itemize}
  \item[(i)] For all $t_0, t_1, t_2 \leq 0$, $\ordexp{\int_{s=t_0}^{t_2} A_s \dd s} = \ordexp{\int_{s=t_1}^{t_2} A_s \dd s}\ordexp{\int_{s=t_0}^{t_1} A_s \dd s}$.
  \item[(ii)] For all $t_0, t_1 \leq 0$, $\ordexp{\int_{s=t_0}^{t_1} A_s \dd s}$ is invertible and $\ordexp{\int_{s=t_0}^{t_1} A_s \dd s}^{-1} = \ordexp{\int_{s=t_1}^{t_0} A_s \dd s}$.
  \item[(iii)] For all $t_0, t \leq 0$, $\det\ordexp{\int_{s=t_0}^t A_s \dd s} = \det(\exp(\int_{s=t_0}^t A_s \dd s)) = \exp(\int_{s=t_0}^t \tr A_s \dd s)$. As a consequence, if $(M_t)_{t \leq 0}$ solves the matrix ordinary differential equation $\dot{M}_t = A_t M_t$, then $m_t = \det M_t$ satisfies $\dot{m}_t = m_t \tr A_t$.
\end{itemize}


\subsection{Solving the backward matrix Riccati equation}\label{app:riccati}

In this section, we solve the backward matrix Riccati equation which corresponds to the second line of~\eqref{eq:Riccati}. In order to alleviate the notations, we drop the superscript notation $x$ on the quantities $\varphi^x$, $Q^x_s$, $R^x_s$, etc., therefore we are led to consider the backward Cauchy problem
\begin{equation}\label{eq:Riccati2}
  \left\{\begin{aligned}
    & \dot{K}_t = K_t^2 + Q_t^\top K_t + K_t Q_t - 2R_t, \quad t < 0,\\
    & \lim_{t \uto 0} K_t^{-1} = 0.
  \end{aligned}\right.
\end{equation}

We employ a quadrature method. Let us first define 
\begin{equation}\label{eq:K0}
  K^0_t = -2\nabla^2 V(\varphi_t).
\end{equation}
By Proposition~\ref{prop:RiccH} and Remark~\ref{rk:forwback}, we have
\begin{equation}\label{eq:K0dot}
  \dot{K}^0_t = \left(K^0_t\right)^2 + Q_t^\top K^0_t + K^0_t Q_t - 2R_t;
\end{equation}
in other words, $K^0_t$ is a particular solution to the backward Riccati equation~\eqref{eq:Riccati2}, but with a different behavior when $t \uto 0$. 

For all $t \leq 0$, we now let 
\begin{equation}\label{eq:A}
  A_t = \nabla^2 V(\varphi_t) + \nabla \ell(\varphi_t) = 2\nabla^2 V(\varphi_t) - Q_t,
\end{equation}
and introduce 
\begin{equation}\label{eq:Z}
  Z_t = \int_{s=t}^0 \OrdExp{\int_{r=s}^t A_r \dd r}\OrdExp{\int_{r=s}^t A_r \dd r}^{\top}\dd s.
\end{equation}
Notice that since $\nabla V$ and $\ell$ are assumed to be smooth, and by definition, the fluctuation path $\varphi^x$ is bounded in $\R^d$, the family of matrices $(A_t)_{t \leq 0}$ is bounded and therefore matches the setting of Section~\ref{app:OE}. Then $Z_t$ solves the time-dependent Lyapunov equation
\begin{equation}\label{eq:ZLyap}
  \dot{Z}_t = -I_d + A_t Z_t + Z_t A_t^\top,
\end{equation}
so that
\begin{equation}\label{eq:ZBer}
  \frac{\dd}{\dd t} Z^{-1}_t = -Z_t^{-1} \dot{Z}_t Z^{-1}_t = (Z^{-1}_t)^2 - A_t^{\top} Z_t^{-1} - Z_t^{-1} A_t.
\end{equation}
We may now define
\begin{equation}\label{eq:Kt}
  K_t = K^0_t + Z_t^{-1},
\end{equation}
for all $t<0$, and check that $\dot{K}_t = K_t^2 + Q_t^\top K_t + K_t Q_t - 2R_t$ using~\eqref{eq:K0}, \eqref{eq:K0dot}, \eqref{eq:A} and \eqref{eq:ZBer}. Furthermore, $K^0_t$ remains bounded when $t \uto 0$ whereas $|t|Z_t^{-1}$ converges to $I_d$, which implies that $K_t^{-1}$ converges to $0$.

As a conclusion, $K_t$ is the solution to~\eqref{eq:Riccati2}, which yields the formula~\eqref{eq:solK} in Section~\ref{ss:FK}.

\begin{remark}\label{rk:newlim}
  It follows from the expression of $K_t$ that $|t|K_t$ converges to $I_d$ when $t \uto 0$. This implies that $|t|^d\det K_t$ converges to $1$, so that the third condition in~\eqref{eq:limit} takes the more explicit form 
  \begin{equation}
    \lim_{t \uto 0} |t|^{-d} \eta_t = 1.
  \end{equation}
\end{remark}


\subsection{Proof of~\eqref{eq:limeta}}\label{app:pfeta}

In this section we prove the identity~\eqref{eq:limeta}. The solution $(K_t)_{t<0}$ to the backward matrix Riccati equation was constructed in the previous section, and it is easily observed that $(\eta_t)_{t<0}$ is defined up to a multiplicative constant by
\begin{equation}\label{eq:etaeta}
  \forall t_1, t_2 < 0, \qquad \frac{\eta_{t_2}}{\eta_{t_1}} = \exp\left(-\int_{s=t_1}^{t_2} \tr K_s \dd s\right).
\end{equation}
The appropriate multiplicative constant shall be chosen in accordance with the third condition of~\eqref{eq:limit}, which shall then provide the correct $t \to -\infty$ limit for $\eta_t$.

We first introduce the notation
\begin{equation}
  U_t = \int_{s=t}^0 \OrdExp{\int_{r=s}^0 A_r \dd r}\OrdExp{\int_{r=s}^0 A_r \dd r}^{\top}\dd s,
\end{equation}
for all $t \leq 0$.

\begin{lemma}\label{lem:etaeta:1}
  For all $t_1, t_2 < 0$,
  \begin{equation}\label{eq:lem:etaeta:1}
    \frac{\eta_{t_2}}{\eta_{t_1}} = \frac{\det U_{t_2}}{\det U_{t_1}}\exp\left(2 \int_{s=t_1}^{t_2} \Delta V(\varphi_s) \dd s\right).
  \end{equation}
\end{lemma}
\begin{proof}
  We first inject the formula~\eqref{eq:Kt} for $K_s$ into~\eqref{eq:etaeta} and obtain
  \begin{equation}
    \frac{\eta_{t_2}}{\eta_{t_1}} = \exp\left(2 \int_{s=t_1}^{t_2} \Delta V(\varphi_s) \dd s - \int_{s=t_1}^{t_2} \tr Z_s^{-1} \dd s\right).
  \end{equation}
  
  By Property~(i) of time ordered exponentials, for all $t<0$,
  \begin{equation}\label{eq:ZtUt}
    Z_t = \OrdExp{\int_{r=0}^t A_r \dd r}U_t\OrdExp{\int_{r=0}^t A_r \dd r}^\top,
  \end{equation}
  so that
  \begin{equation}
    \OrdExp{\int_{r=t}^0 A_r \dd r}^\top = Z_t^{-1}\OrdExp{\int_{r=0}^t A_r \dd r}U_t,
  \end{equation}
  where we have used Property~(ii) of time ordered exponentials. On the other hand, 
  \begin{equation}
    \dot{U}_t = -\OrdExp{\int_{r=t}^0 A_r \dd r}\OrdExp{\int_{r=t}^0 A_r \dd r}^{\top},
  \end{equation}
  therefore
  \begin{equation}
    \dot{U}_t = -\OrdExp{\int_{r=t}^0 A_r \dd r}Z_t^{-1}\OrdExp{\int_{r=0}^t A_r \dd r}U_t
  \end{equation}
  and by Property~(iii) of time ordered exponentials, $m_t = \det U_t$ satisfies
  \begin{equation}
    \dot{m}_t = - m_t\tr\left(\OrdExp{\int_{r=t}^0 A_r \dd r}Z_t^{-1}\OrdExp{\int_{r=0}^t A_r \dd r}\right)
  \end{equation}
  which reduces to
  \begin{equation}
    \dot{m}_t = - m_t\tr Z_t^{-1}
  \end{equation}
  thanks to Property~(ii) of time ordered exponentials again. As a consequence,
  \begin{equation}
    \exp\left(-\int_{s=t_1}^{t_2} \tr Z_s^{-1} \dd s\right) = \frac{m_{t_2}}{m_{t_1}} = \frac{\det U_{t_2}}{\det U_{t_1}},
  \end{equation}
  which completes the proof.
\end{proof}

In the next lemma we describe the $t_2 \uto 0$ limit of the ratio $\eta_{t_2}/\eta_{t_1}$.
\begin{lemma}
  For all $t_1 < 0$,
  \begin{equation}
    \eta_{t_1} = \det U_{t_1} \exp\left(-2 \int_{s=t_1}^0 \Delta V(\varphi_s) \dd s\right).
  \end{equation}
\end{lemma}
\begin{proof}
  We fix $t_1<0$ and let $t_2$ grow to $0$ in~\eqref{eq:lem:etaeta:1}. By the definition of $U_t$, $|t_2|^{-1}U_{t_2}$ converges to $I_d$, so that by Remark~\ref{rk:newlim},
  \begin{equation}
    \lim_{t_2 \uto 0} \frac{\eta_{t_2}}{\det U_{t_2}} = 1,
  \end{equation}
  which completes the proof.
\end{proof}

We finally address the $t_1 \to -\infty$ limit of the expression obtained for $\eta_{t_1}$.

\begin{lemma}
  We have
  \begin{equation}
    \lim_{t_1 \to -\infty} \eta_{t_1} = \frac{1}{2^d \det(\nabla^2 V(\bar{x}))}\exp\left(2\int_{s=-\infty}^0 \dive \ell(\varphi_s)\dd s\right).
  \end{equation}
\end{lemma}
\begin{proof}
  Let $t_1<0$. Using Properties~(i) and~(iii) of time ordered exponentials and~\eqref{eq:ZtUt}, we rewrite
  \begin{equation}
    \begin{aligned}
      \eta_{t_1} & = \det\left(\exp\left(-\int_{s=t_1}^0 \nabla^2 V(\varphi_s)\dd s\right) U_{t_1} \exp\left(-\int_{s=t_1}^0 \nabla^2 V(\varphi_s)\dd s\right)\right)\\
      & = \exp\left(2 \int_{s=t_1}^0 \dive \ell(\varphi_s)\dd s\right) \det\left(\exp\left(\int_{s=t_1}^0 A_s \dd s\right) U_{t_1} \exp\left(\int_{s=t_1}^0 A_s \dd s\right)^\top\right)\\
      & = \exp\left(2 \int_{s=t_1}^0 \dive \ell(\varphi_s)\dd s\right) \det Z_{t_1}.
    \end{aligned}
  \end{equation}
  We are therefore led to compute the $t_1 \to -\infty$ limit of $Z_{t_1}$. To this aim we recall that $(Z_t)_{t \leq 0}$ solves the time-dependent Lyapunov equation~\eqref{eq:ZLyap}. In addition to the notation $\bar{H} = \nabla^2 V(\bar{x})$ introduced in Section~\ref{ss:barH}, let us denote $\bar{D} = \nabla \ell(\bar{x})$, so that taking the $t \to -\infty$ limit of~\eqref{eq:ZLyap} shows that
  \begin{equation}
    \bar{Z} = \lim_{t \to -\infty} Z_t
  \end{equation}
  satisfies the stationary Lyapunov equation
  \begin{equation}
    (\bar{H}+\bar{D})\bar{Z} + \bar{Z}(\bar{H}+\bar{D})^\top = I_d.
  \end{equation}
  In addition,
  \begin{itemize}
    \item evaluating the identity~\eqref{eq:D2transv} at $\bar{x}$ yields $\bar{D}^\top \bar{H} + \bar{H}\bar{D} = 0$,
    \item $\bar{H}$ is assumed to be positive-definite.
  \end{itemize}
  As a consequence, 
  \begin{equation}
    \bar{Z} = \frac{1}{2} \bar{H}^{-1},
  \end{equation}
  from which we deduce that
  \begin{equation}
    \lim_{t_1 \to -\infty} \eta_{t_1} = \exp\left(2 \int_{s=-\infty}^0 \dive \ell(\varphi_s)\dd s\right) \det\left(\frac{1}{2} \bar{H}^{-1}\right),
  \end{equation}
  and the proof is completed.
\end{proof}

\begin{acknowledgements}
  We would like to thank A. Alfonsi, A. Levitt, G. Stoltz for useful comments on the numerical integration of matrix Riccati equations. We also thank two anonymous referees for their careful reading of the paper which helped improving the presentation of our results.
    
  The research leading to these results has received funding from the European Research Council under the European Union's seventh Framework Programme (FP7/2007-2013 Grant Agreement No. 616811). During the last stage of this work, F. Bouchet received support by a subagreement from the Johns Hopkins University with funds provided by Grant No. 663054 from Simons Foundation. Its contents are solely the responsibility of the authors and do not necessarily represent the official views of Simons Foundation or the Johns Hopkins University. J. Reygner is supported by the French National Research Agency (ANR) under the programs EFI (ANR-17-CE40-0030) and QuAMProcs (ANR-19-CE40-0010).
\end{acknowledgements}

\section*{Data availability and conflict of interest statement}
This theoretical work uses no external dataset. 

The authors have no relevant financial or non-financial interests to disclose, have no competing interests to declare that are relevant to the content of this article. The authors declare no conflict of interest relevant to the content of this article.



\begin{thebibliography}{10}
\providecommand{\url}[1]{{#1}}
\providecommand{\urlprefix}{URL }
\expandafter\ifx\csname urlstyle\endcsname\relax
  \providecommand{\doi}[1]{DOI~\discretionary{}{}{}#1}\else
  \providecommand{\doi}{DOI~\discretionary{}{}{}\begingroup
  \urlstyle{rm}\Url}\fi

\bibitem{abbot2021rare}
Abbot, D.S., Webber, R.J., Hadden, S., Weare, J.: Rare event sampling improves
  mercury instability statistics.
\newblock arXiv preprint arXiv:2106.09091  (2021)


\bibitem{Ber13}
Berglund, N.: Kramers' law: validity, derivations and generalisations. 
\newblock Markov Process. Related Fields \textbf{19}(3), 459--490 (2013)

\bibitem{BerDeSGabJonLan15}
L. Bertini, A. De Sole, D. Gabrielli, G. Jona-Lasinio, and C. Landim.: Macroscopic fluctuation theory.
\newblock Rev. Modern Phys. \textbf{87}, 593--636 (2015)

\bibitem{BouNarGaw16}
Bouchet, F., Nardini, C., and Gawedzki, K.: Perturbative calculation of quasi-potential in non-equilibrium diffusions: a mean-field example.
\newblock J. Stat. Phys. \textbf{163}, 1157--1210 (2016)

\bibitem{BouRey16}
Bouchet, F. and Reygner, J.: Generalisation of the Eyring--Kramers transition rate
  formula to irreversible diffusion processes.
\newblock Ann. Henri Poincar\'e \textbf{17}(12), 3499--3532 (2016)

\bibitem{bouchet2019rare}
Bouchet, F., Rolland, J., Simonnet, E.: Rare event algorithm links transitions
  in turbulent flows with activated nucleations.
\newblock Physical review letters \textbf{122}(7), 074,502 (2019)

\bibitem{BovEckGayKle04}
Bovier, A., Eckhoff, M., Gayrard, V., and Klein, M.: Metastability in reversible diffusion processes. I. Sharp asymptotics for capacities and exit times.
\newblock J. Euro. Math. Soc. \textbf{6}(4), 399--424 (2004)

\bibitem{Callan:1977pt}
Callan Jr., C.G., Coleman, S.R.: {The Fate of the False Vacuum. 2. First
  Quantum Corrections}.
\newblock Phys. Rev. D \textbf{16}, 1762--1768 (1977).
\newblock \doi{10.1103/PhysRevD.16.1762}

\bibitem{CohLew67}
Cohen, J.K. and Lewis, R.M.: A ray method for the asymptotic solution of the
  diffusion equation.
\newblock IMA J. Appl. Math. \textbf{3}(3), 266--290 (1967)


\bibitem{Coleman:1978ae}
Coleman, S.R.: {The Uses of Instantons}.
\newblock Subnucl. Ser. \textbf{15}, 805 (1979).
\newblock [,382(1978)]


\bibitem{dematteis2018rogue}
Dematteis, G., Grafke, T., Vanden-Eijnden, E.: Rogue waves and large deviations
  in deep sea.
\newblock Proceedings of the National Academy of Sciences \textbf{115}(5),
  855--860 (2018)


\bibitem{DemZei10}
Dembo, A. and Zeitouni, O.: Large deviations techniques and applications,
  Stochastic Modelling and Applied Probability, vol.~38.
\newblock Springer-Verlag, Berlin (2010).
\newblock Corrected reprint of the second edition.

\bibitem{DieEir94}
Dieci, L. and Eirola, T.: Positive definiteness in the numerical solution of Riccati differential equations.
\newblock Numer. Math. \textbf{67}, 303--313 (1994)

\bibitem{DieEir96}
Dieci, L. and Eirola, T.: Preserving monotonicity in the numerical solution of Riccati differential equations.
\newblock Numer. Math. \textbf{74}, 35--47 (1996)

\bibitem{DubSai00}
Dubois, F. and Sa\"idi, A.: Unconditionnally stable scheme for Riccati equation.
\newblock ESAIM Proc. \textbf{8}, 39--52 (2000)

\bibitem{FerGra}
Ferr\'e, G. and Grafke, T.: Approximate optimal controls via instanton expansion for low temperature free energy computation.
\newblock Preprint arXiv:2011.10990.


\bibitem{FreWen12}
Freidlin, M.I. and Wentzell, A.D.: Random perturbations of dynamical systems, Grundlehren der Mathematischen Wissenschaften, vol.~260. 
\newblock Springer, Heidelberg (2012).
\newblock Translated from the 1979 Russian original by Joseph Sz\"ucs. Third edition.



\bibitem{grafke2013instanton}
Grafke, T., Grauer, R., Sch{\"a}fer, T.: Instanton filtering for the stochastic
  burgers equation.
\newblock Journal of Physics A: Mathematical and Theoretical \textbf{46}(6),
  062,002 (2013)

\bibitem{grafke2015efficient}
Grafke, T., Grauer, R., Schindel, S.: Efficient computation of instantons for
  multi-dimensional turbulent flows with large scale forcing.
\newblock Communications in Computational Physics \textbf{18}(3), 577--592
  (2015)

\bibitem{GraSchVan}
Grafke, T., Sch\"afer, T., Vanden-Eijnden, E.: Sharp Asymptotic Estimates for Expectations, Probabilities, and Mean First Passage Times in Stochastic Systems with Small Noise.
\newblock Preprint arXiv:2103.04837.


\bibitem{grafke2019numerical}
Grafke, T., Vanden-Eijnden, E.: Numerical computation of rare events via large
  deviation theory.
\newblock Chaos: An Interdisciplinary Journal of Nonlinear Science
  \textbf{29}(6), 063,118 (2019)

\bibitem{Gra88}
Graham, R.: Macroscopic potentials, bifurcations and noise in dissipative systems. 
\newblock In Noise in Nonlinear Dynamical Systems, 1, 225--278. Cambridge University Press (1988)

\bibitem{HeyVan08:CPAM}
Heymann, M. and Vanden-Eijnden, E.: The geometric minimum action method: A least
  action principle on the space of curves.
\newblock Comm. Pure Appl. Math. \textbf{61}(8), 1052--1117 (2008)

\bibitem{kampen_stochastic_2007}
Kampen, N.G.v.: Stochastic processes in physics and chemistry, 3rd ed edn.
\newblock North-{Holland} personal library. Elsevier, Amsterdam ; Boston (2007)


\bibitem{LanMarSeo19}
Landim, C., Mariani, M., and Seo, I.: Dirichlet's and {T}homson's principles for non-selfadjoint elliptic
  operators with application to non-reversible metastable diffusion processes.
\newblock {Arch. Ration. Mech. Anal.}, 231(2):887--938, 2019.

\bibitem{langer2006excitation}
Langer, J.: Excitation chains at the glass transition.
\newblock Physical review letters \textbf{97}(11), 115,704 (2006)

\bibitem{langer_1967_condensation_point}
{Langer}, J.S.: {Theory of the condensation point}.
\newblock Annals of Physics \textbf{41}, 108--157 (1967).
\newblock \doi{10.1016/0003-4916(67)90200-X}

\bibitem{LAURIE:2015:A}
Laurie, J., Bouchet, F.: {Computation of rare transitions in the barotropic
  quasi-geostrophic equations}.
\newblock {NEW JOURNAL OF PHYSICS} \textbf{{17}} ({2015}).
\newblock \doi{{10.1088/1367-2630/17/1/015009}}

\bibitem{LeeSeo}
Lee, J. and Seo, I.: Non-reversible metastable diffusions with {G}ibbs invariant measure
  {I}: {E}yring-{K}ramers formula.
\newblock Preprint arXiv:2008.08291.


\bibitem{LuStuWeb16}
Lu, Y., Stuart, A.M., and Weber, H.: Gaussian approximations for transition paths
  in molecular dynamics.
\newblock SIAM J. Math. Anal. \textbf{49}(4), 3005--3047 (2017)

\bibitem{Lud75}
Ludwig, D.: Persistence of dynamical systems under random perturbations.
\newblock Siam Rev. \textbf{17}(4), 605--640 (1975)

\bibitem{MaiSte97}
Maier, R.S. and Stein, D.L.: Limiting exit location distributions in the
  stochastic exit problem.
\newblock SIAM J. Appl. Math. \textbf{57}(3), 752--790 (1997)

\bibitem{PasCam}
Paskal, N. and Cameron, M.: An efficient jet marcher for computing the quasipotential for 2D SDEs.
\newblock To appear in J. Sci. Comput.

\bibitem{RAGONE:2018:A}
Ragone, F., Wouters, J., Bouchet, F.: {Computation of extreme heat waves in
  climate models using a large deviation algorithm}.
\newblock {PROCEEDINGS OF THE NATIONAL ACADEMY OF SCIENCES OF THE UNITED STATES
  OF AMERICA} \textbf{{115}}({1}), {24--29} ({2018}).
\newblock \doi{{10.1073/pnas.1712645115}}

\bibitem{SanStu16}
Sanz-Alonso, D. and Stuart, A.M.: Gaussian approximations of small noise
  diffusions in Kullback-Leibler divergence.
\newblock Comm. Math. Sci. \textbf{15}(7), 2087--2097 (2017)

\bibitem{Sch09}
Schuss, Z.: Theory and applications of stochastic processes: an analytical
  approach, Applied Mathematical Sciences, vol.~170.
\newblock Springer, New York (2010).


\bibitem{simonnet2021multistability}
Simonnet, E., Rolland, J., Bouchet, F.: Multistability and rare spontaneous
  transitions in barotropic $\beta$-plane turbulence.
\newblock Journal of the Atmospheric Sciences \textbf{78}(6), 1889--1911 (2021)

\bibitem{VanHey08:JCP}
Vanden-Eijnden and E., Heymann, M.: The geometric minimum action method for
  computing minimum energy paths.
\newblock J. Chem. Phys. \textbf{128}(6), 061--103 (2008)

\bibitem{woillez2020instantons}
Woillez, E., Bouchet, F.: Instantons for the destabilization of the inner solar
  system.
\newblock Physical Review Letters \textbf{125}(2), 021,101 (2020)

\bibitem{woillez2019escape}
Woillez, E., Zhao, Y., Kafri, Y., Lecomte, V., Tailleur, J.: Activated escape
  of a self-propelled particle from a metastable state.
\newblock Phys. Rev. Lett. \textbf{122}, 258,001 (2019).
\newblock \doi{10.1103/PhysRevLett.122.258001}.
\newblock
  \urlprefix\url{https://link.aps.org/doi/10.1103/PhysRevLett.122.258001}

\bibitem{zinn1996quantum}
Zinn-Justin, J.: Quantum field theory and critical phenomena.
\newblock Clarendon Press (1996)

\end{thebibliography}
\end{document}